\documentclass{article}
\usepackage{epstopdf}
\usepackage{amsmath}
\usepackage{amsthm}
\usepackage{amssymb}
\usepackage{amsfonts}
\usepackage{mathrsfs,bbm}
\usepackage{mathrsfs}
\usepackage[all]{xy}
\usepackage{color}
\usepackage{pdfcolmk}
\usepackage{comment}

\numberwithin{equation}{section}

\theoremstyle{plain}
\newtheorem{theorem}{Theorem}

\newtheorem{corollary}{Corollary}

\newtheorem{lemma}{Lemma}

\newtheorem{proposition}{Proposition}

\theoremstyle{definition}
\newtheorem{definition}{Definition}

\theoremstyle{remark}
\newtheorem{example}{Example}
\newtheorem{remark}{Remark}

\renewcommand{\1}{\mathbbm{1}}

\DeclareFontFamily{OT1}{wncyi}{}
\DeclareFontShape{OT1}{wncyi}{m}{it}{
<5> <6> <7> <8> <9> gen * wncyi
<10> <10.95> <12> <14.4> <17.28> <20.74> <24.88> wncyi10
}{}
\DeclareSymbolFont{cyrletters}{OT1}{wncyi}{m}{it}
\DeclareSymbolFontAlphabet{\cyrmath}{cyrletters}
\DeclareMathSymbol{\rE}{\cyrmath}{cyrletters}{003}
\DeclareMathSymbol{\rD}{\cyrmath}{cyrletters}{068}
\DeclareMathSymbol{\rG}{\cyrmath}{cyrletters}{017}
\DeclareMathSymbol{\rI}{\cyrmath}{cyrletters}{073}
\DeclareMathSymbol{\rL}{\cyrmath}{cyrletters}{076}
\DeclareMathSymbol{\rZ}{\cyrmath}{cyrletters}{090}
\newcommand{\ev}{\mathrm{ev}}

\renewcommand{\phi}{\varphi}
\newcommand{\Hc}{\mathcal{H}}

\newcommand{\Oc}{\mathcal{O}}
\newcommand{\Ic}{\mathcal{I}}

\newcommand{\Dc}{\mathcal{D}}
\newcommand{\Bc}{\mathcal{B}}
\newcommand{\Mc}{\mathcal{M}}
\newcommand{\Nc}{\mathcal{N}}

\newcommand{\Lc}{\mathcal{L}}

\newcommand{\Sc}{\mathcal{S}}
\newcommand{\Kc}{\mathcal{K}}
\newcommand{\Pc}{\mathcal{P}}
\newcommand{\Ec}{\mathcal{E}}

\newcommand{\un}{\mathrm{un}}
\newcommand{\as}{\mathrm{as}}
\newcommand{\com}{\mathrm{com}}

\newcommand{\Fc}{\mathcal{F}}

\newcommand{\Rc}{\mathcal{R}}

\newcommand{\Cc}{\mathcal{C}}
\newcommand{\Uc}{\mathcal{U}}

\newcommand{\Ac}{\mathcal{A}}

\newcommand{\Db}{\mathbb{D}}
\newcommand{\A}{\mathbb{A}}
\newcommand{\Lb}{\mathbb{L}}
\newcommand{\C}{\mathbb{C}}
\newcommand{\R}{\mathbb{R}}

\newcommand{\Z}{\mathbb{Z}}

\newcommand{\N}{\mathbb{N}}

\newcommand{\mfk}{\mathfrak{m}}

\newcommand{\gfk}{\mathfrak{g}}

\newcommand{\Hb}{\mathbb{H}}

\newcommand{\Tc}{\mathcal{T}}

\newcommand{\Cliff}{\mathrm{Cliff}}

\newcommand{\Conn}{\mathrm{Conn}}
\newcommand{\Dirac}{\displaystyle{\not}D}

\newcommand{\uSpec}{\underline{\mathrm{Spec}}}
\newcommand{\Spec}{\mathrm{Spec}}
\newcommand{\Bun}{\mathrm{Bun}}

\newcommand{\Ker}{\mathrm{Ker}}

\newcommand{\Mod}{\mathrm{Mod}}

\newcommand{\Abrm}{\mathrm{Ab}}
\newcommand{\Sym}{\mathrm{Sym}}

\newcommand{\Sol}{\mathrm{Sol}}
\newcommand{\Open}{\mathrm{Open}}
\newcommand{\Groupoids}{\mathrm{Groupoids}}
\newcommand{\Cov}{\mathrm{Cov}}

\newcommand{\Hom}{\mathrm{Hom}}
\newcommand{\Ber}{\mathrm{Ber}}

\newcommand{\uHom}{\underline{\mathrm{Hom}}}
\newcommand{\uGamma}{\underline{\Gamma}}
\newcommand{\uExt}{\underline{\mathrm{Ext}}}
\newcommand{\uDiff}{\underline{\mathrm{Diff}}}

\newcommand{\Lotimes}{\overset{\Lb}{\otimes}}

\renewcommand{\1}{\mathbbm{1}}
\newcommand{\FSets}{\mathrm{FSets}}

\newcommand{\Homc}{\mathcal{H}om}
\newcommand{\Extc}{\mathcal{E}xt}

\newcommand{\Legos}{\textsc{Legos}}

\newcommand{\Der}{\mathrm{Der}}
\newcommand{\uDer}{\underline{\mathrm{Der}}}
\newcommand{\Derc}{\mathcal{D}er}
\newcommand{\Gr}{\mathrm{Gr}}

\newcommand{\Jet}{\mathrm{Jet}}
\newcommand{\Aff}{\textsc{Aff}}
\newcommand{\Sch}{\textsc{Sch}}
\newcommand{\Alg}{\textsc{Alg}}
\newcommand{\Vect}{\textsc{Vect}}

\newcommand{\Sets}{\textsc{Sets}}

\newcommand{\SSets}{\textsc{SSets}}

\newcommand{\End}{\mathrm{End}}

\newcommand{\id}{\mathrm{id}}

\newcommand{\gr}{\mathrm{gr}}
\newcommand{\sgn}{\mathrm{sgn}\;}

\newcommand{\Lie}{\mathrm{Lie}}

\newcommand{\limind}{{\underset{\longrightarrow}{\mathrm{lim}}\,}}

\newcommand{\Sh}{\mathrm{Sh}}
\newcommand{\DR}{\mathrm{DR}}

\newcommand{\Langle}{\big\langle}
\newcommand{\Rangle}{\big\rangle}

\title{Histories and observables\\
in\\
covariant field theory}

\author{Fr\'ed\'eric Paugam}

\begin{document}
\maketitle
%\begin{frontmatter}

%\ead{frederic.paugam@math.jussieu.fr}
%\address{Universit\'e Paris 6\\
%Institut de math\'ematiques de Jussieu\\
%175, rue du Chevaleret, 75013 Paris}

\begin{abstract}
Motivated by DeWitt's viewpoint of covariant field theory,
we define a general notion of non-local classical
observable that applies to many physical Lagrangian
systems (with bosonic and fermionic variables), by using
methods that are now standard in algebraic geometry.
We review the methods of local functional
calculus, as they are presented by Beilinson and Drinfeld,
and relate them to our construction. We partially explain the relation
of these with Vinogradov's secondary calculus.
The methods present here are all necessary to understand
mathematically properly and with simple notions the full renormalization
of the standard model, based on functional integral methods.
Our approach is close in spirit to non-perturbative methods since
we work with actual functions on spaces of fields, and not only
formal power series.
This article can be seen as an introduction to well-grounded
classical physical mathematics, and as a good starting point
to study quantum physical mathematics, which make frequent use
of non-local functionals, like for example in the computation
of Wilson's effective action.
We finish by describing briefly a coordinate-free approach to
the classical Batalin-Vilkovisky formalism for general gauge theories,
in the language of homotopical geometry.
\end{abstract}

%\end{frontmatter}

%\maketitle

\newpage
\tableofcontents
\newpage

%************************************************************************
\section{Introduction}

\begin{flushright}
\emph{We [physicists] often do not know a priori\\
just where a given formalism [...] will take us.\\
We are compelled to leave it to the mathematicians\\
to tidy things up after we have left the playing field.\\
Bryce S. DeWitt.\\
The global approach to quantum field theory.}
\end{flushright}

A recurrent difficulty that the average mathematician has to face if
he opens an experimental physics book is the scarcity of definitions for the
mathematical objects used in computations. This does not mean
that the computations of experimental physicists are false, only that they are hard
to approach from a mathematician's viewpoint.

The aim of this article is to show
how standard abstract methods of pure mathematics (functors, sheaves, homotopical
spaces, monoidal categories and operads) can
be seen as parts of applied mathematics, since they are all necessary to
explain to a mathematician the standard model of particle physics or supersymmetric field theory
without coordinates.
The main advantage of computations without coordinates
is that they are often very algebraic, and general enough to be used in very
different contexts. Moreover, they open the road to many interesting mathematical
applications.

The starting point of particle physics is the study of variational
problems, their symmetries and conservation laws.
There are at least three ways to study such problems: functional analysis,
local functional calculus and non-local functional calculus.

The functional analytic methods (analysis on spaces of functions) proved
to be very useful in the linear case but have a quite limited scope for non linear problems.

Local functional calculus on jet spaces has been developed by many authors, starting
from Noether and her famous theorems (see \cite{Kosmann-Schwarzbach}), and including
Gelfand, Manin, Vinogradov and Gromov. We will mainly be interested
in Vinogradov's  $\Cc$-spectral sequence \cite{Vinogradov} and secondary
calculus, studied in the smooth setting in various reference books by
many people (see \cite{Vinogradov}, \cite{Vinogradov-Krasilshchik},
\cite{Krasilshchik-Verbovetsky-1998}),
and the algebraic approach of Beilinson and Drinfeld \cite{Beilinson-Drinfeld-Chiral},
using the language of differential algebra, which originated in Ritt's school
\cite{Ritt}. To sum up, these methods allow one to compute
symmetries and conservation laws of partial differential equations systematically and to solve
the inverse problem of variational calculus algebraically. It also allows one to prove
deep results on partial differential equations like the h-principle \cite{Gromov}, but we will
not go into this since we are only interested in the very formal computations
of physicists.

Non-local functional calculus is all around in the physical literature,
and its mathematical formalization is very close to Grothendieck's functorial approach
of geometry \cite{SGA4}. It was first introduced in physics in a special
case by Souriau \cite{Souriau} and his school \cite{piz} (see also \cite{classical-fields-deligne-freed}).
It will prove very close to the physicists' way of thinking of variational problems.
This functorial approach is already used in finite dimension in the IAS lectures
\cite{notes-on-supersymmetry}, and for spaces of fields in Lott's article
\cite{Lott-torsion-supergeometry}, but was not systematically developed
mathematically there (no sheaf condition, for example).

These three methods have their advantages and drawbacks.
Since the literature on functional methods is already very large,
we will concentrate on the two other approaches.
One can compare these two last theories to Grothendieck's scheme theory in the
setting of partial differential equations: it does not allow us to solve the equations explicitly
(this can be done, anyway, only in very special cases), but it gives powerful
methods to define and compute some very interesting invariants associated to
the geometry of the space of solutions. These methods
are based on differential calculus on the ``space of solutions'' of the given partial differential equations. The local and non-local calculus give two useful definitions
of this notion of ``space of solution'', which is central in modern physics.

The algebro-geometric methods are based on the ``punctual''
approach to geometry, which is essentially the way physicists think:
one does not have to wonder exactly on which functional space
one works, the only important things being what a parameterized function
is, and the formal changes of parameterizations that are allowed
(i.e., the categories and morphisms in play).
These methods already proved to be very useful in understanding
fermionic differential and integral calculus and they give a geometrically
intuitive way to work with spaces of solutions of nonlinear
partial differential equations. Our methods are aimed at the study of
non-perturbative quantum field theory, but here we describe here mainly
classical field theory.

We finish by briefly describing a general coordinate-free approach to
the classical Batalin-Vilkovisky formalism for general gauge theories.

In this paper, all classical manifolds used are implicitly supposed to be oriented,
if one needs to integrate differential forms on them.

\begin{comment}
The aim of this document is to study algebro-geometric versions of the theory
of diffieties and secondary calculus. In their big opus \cite{Beilinson-Drinfeld-Chiral}
on Chiral Algebras, Chapter 2, Beilinson and Drinfeld (and independently, Malgrange
in \cite{Malgrange-D-Variete}, Chapter IV) introduced an algebraic version of diffieties,
called $\Dc$-schemes. Their construction is based
on relative geometry on the tensor category of modules with connections, and
this geometry was actually already introduced by Deligne in his work \cite{De}
on geometry in tannakian categories.
The recent developments of relative (see \cite{Toen-Vaquie},
\cite{Durov-2007}) and derived (see \cite{Toen-Vezzosi-HAG-DAG}) algebraic geometry
give new tools to study $\Dc$-schemes quite systematically.

From this algebraic point of view, the algebraic description
of the (unrenormalized) covariant quantization of a given classical system is reduced
to the problem of defining a category of generalized algebras in
which one can write the functor of solutions
$$
\mathrm{DS}:A\mapsto  \left\{Z(J)\in A,\,
\left[\frac{\delta S}{\delta \phi(x)}\left(\frac{\delta}{\delta J(x)}\right)-J(x)\right].Z(J)=0
\right\}
$$
of the Dyson-Schwinger equation. The resolution of this problem is out of the scope
of this paper but can be seen as another motivating background for studying classical
field theory from the functorial point of view.
\end{comment}

%************************************************************************
\section{Points, coordinates and histories in geometry and physics}
There are two complementary ways of studying spaces in geometry and physics.
\begin{enumerate}
\item The \emph{functional viewpoint}, based on the notion of a \emph{coordinate} function,
also called observable
(see Nestruev \cite{Nestruev}, Connes \cite{Connes2} and most
of the literature on mathematical physics), translates most of the geometrical
constructions in algebraic terms.
\item The \emph{punctual viewpoint}, based on the notion of a \emph{point}
(see \cite{Souriau}, \cite{piz} and \cite{Gro1}), studies a given
space by giving all its parameterized families of points with values in some given
standard building blocks, like open subsets of $\R^n$ for example.
\end{enumerate}

Both of these methods have their advantages.
The main objects of modern physics, called spaces of histories or spaces of fields,
are functional spaces. In studying such infinite-dimensional spaces,
the functional viewpoint becomes sometimes too cumbersome or even
inefficient, and the systematic use of the punctual viewpoint proves to be closer
to the physicists' (sometimes informal) language. Various combinations of both of them
will give optimal geometrical contexts for physics.

%************************************************************************
\subsection{Lagrangian variational problems}

We first give a definition of a lagrangian variational problem,
which is general enough to treat many variational problems that appear in classical and
quantum physics (classical mechanics, Yang-Mills theory, general relativity, fermionic
field theory and supersymmetric sigma models).
We base our approach on a general notion of space, which will be described in this paper.

\begin{definition}
A lagrangian variational problem is composed of the following data:
\begin{enumerate}
\item a space $M$ called the parameter space for trajectories,
\item a space $C$ called the configuration space for trajectories,
\item a morphism $\pi:C\to M$ (often supposed to be surjective),
\item  a subspace $H\subset \Gamma(M,C)$ of the space of sections of $\pi$
$$\Gamma(M,C):=\{x:M\to C,\;\pi\circ x=\id\},$$
called the space of histories, and
\item a functional $S:H\to A$ (where $A$ is a space in rings that is often the real line
$\R$ or $\R[[\hbar]]$)
called the action functional.
\end{enumerate}
The space of classical trajectories for the variational problem is
the subspace $T$ of $H$ defined by
$$T=\{x\in H|\;d_{x}S=0\}.$$
If $B$ is another space, a classical $B$-valued observable is a functional
$F:T\to B$ and a quantum $B$-valued observable is a functional
$F:H\to B$.
\end{definition}

We will now give some physical examples, without going into details.
The definitions of the types of space that are
necessary to formalize these examples properly in the above language arise
with an increasing level of difficulties.

In classical mechanics, the parameter space $M$ for trajectories is
a compact time interval, e.g. $M=[0,1]$, the configuration bundle
is the natural projection $\pi:C=[0,1]\times \R^3\to [0,1]=M$,
and the space of histories is the space of trajectories $x:[0,1]\to \R^3$
with some fixed starting and ending points $x_{0}$ and $x_{1}$.
The action functional of the free particle is given by the formula
$$S(x)=\int_M \frac{1}{2}m\|\partial_{t}x\|^2dt.$$

To describe the variational problem of pure Yang-Mills theory, one
needs a lorentzian spacetime manifold $(M,g)$ and a principal $G$-bundle $P$ on $M$.
The projection $\pi:C=\Conn_{G}(P)\to M$ is the bundle whose sections $A$ are
principal $G$-connections ($G$-equivariant covariant derivatives).
The action functional is given by the formula
$$S(A)=\int_{M}F_{A}\wedge *F_{A}$$
where $F_{A}$ is the curvature of the connection $A$ on $P$.

In fermionic field theory, one starts as before with a four-dimensional
lorentzian manifold $(M,g)$ and a spinor bundle $S\to M$.
The projection $\pi:C=\Pi S\to M$ is the odd fiber bundle associated to
$S$ (with fiber odd supervector spaces) and the action functional is given by 
$$S(\psi)=\int_{M} \langle \psi,\Dirac\psi\rangle dx$$
where the above pairing is described in \cite{notes-on-supersymmetry}.
It is already hard there to give a proper sense to this expression,
because a usual section $\psi:M\to \Pi S$ is essentially trivial.

In supersymmetric sigma models, one starts with a supervariety (see
below for a general definition), for example $\R^{1|1}$,
and works with the bundle $\pi:C=\R\times \R^{1|1}\to \R^{1|1}$.
The supersymmetric lagrangian is then given by a superintegral over
$\R^{1|1}$.

To be able to treat all these examples on an equal footing, we will need
a flexible enough notion of ``space'' of trajectories and histories.

%************************************************************************
\subsection{Points and coordinates: useful nonsense}
To define a very general notion of space, one needs a category
of geometrical building blocks, which we call $\Legos$. This must be equipped
with a (subcanonical) Grothendieck topology $\tau$.
\begin{comment}
This data gives rise to that of a topos (of its set valued sheaves),
and the language of topoi (for which a nice reference is \cite{Maclane-Moerdijk},
see also \cite{Cartier}) is sometimes more
practical to prove theorems avoiding technical details. Since our aim is to communicate
with physicists, we prefer to keep track of the category $\Legos$ of building blocks, since
these are objects better understandable by the average scientist.
\end{comment}

Recall that if $\Legos$ is a category, and if we denote $\Legos^\vee$ the category
of contravariant functors from $\Legos$ to $\Sets$, there is a
fully faithful Yoneda embedding
$$
\begin{array}{ccc}
\Legos & \to 		& \Legos^\vee\\
X          &\mapsto & \underline{X}:=\Hom(.,X):\Legos\to \Sets.
\end{array}
$$
This gives an embedding of $\Legos$ in a category that contains
all limits and colimits.

If one starts with the opposite category
$\Legos^{op}$ instead of $\Legos$, the ioneda embedding
will be given by
$$
\begin{array}{ccc}
\Legos^{op} 	& \to 		& (\Legos^{op})^\vee\\
X          		&\mapsto & \underline{X}:=\Hom_{\Legos^{op}}(.,X)=\Hom_{\Legos}(X,.).
\end{array}
$$

The above constructions are the categorical counterparts of the two
viewpoint of spaces used by physicists:
\begin{enumerate}
\item in the punctual viewpoint, points of a space $X$ with values in legos are given
by the contravariant functor $\Hom(.,X):\Legos\to \Sets$;
\item in the functional viewpoint, coordinates on a space $X$ are given by the
covariant functor
$\Hom(X,.):\Legos\to \Sets$, and as seen above, these are also points
of $X$ with values in objects of $\Legos^{op}$.
\end{enumerate}

Here are the main advantages of the two approaches.
\begin{enumerate}
\item The punctual viewpoint is very natural for the study of contravariant
functorial constructions (i.e. constructions with a pull-back), like differential forms for example.
\item The functional viewpoint is very natural for the study of covariant
constructions like vector fields and differential operators.
\end{enumerate}

It will thus often be useful to combine those two approaches by using
as $\Legos$ categories of algebras of functions, for which one
can apply covariant and contravariant constructions.
\begin{comment}
One could also
use bifunctors of the form
$$\Hom_{\Legos_{1}}(.,X)\times \Hom_{\Legos_{2}}(X,.),$$
verfying some compatibility conditions to get something like
a generalized ringed topos (see \cite{Lurie-DAG-V}).
\end{comment}

If one works with the category $\Legos^\vee$, one gets into trouble when pasting
building blocks. Indeed, suppose for example that $\Legos=\Open$ is the category
of open subsets of $\R^n$ for varying $n$. Then if an open set $U$ is covered by
two open subsets, i.e.,
$$U=U_{1}\cup U_{2}=:U_{1}\coprod_{U}U_{2},$$
one usually does not have
$$\underline{U}=\underline{U_{1}}\coprod_{\underline{U}}\underline{U_{2}}$$
in the category $\Legos^\vee$. One will thus get into trouble if one wants do define
varieties by pasting the spaces associated to legos.
The sheaf condition is here exactly to prevent this bad situation happening.

Spaces will thus simply be the category $\Sh(\Legos,\tau)\subset \Legos^\vee$ of sheaves
of sets on $\Legos$ with respect to the topology $\tau$. Yoneda's lemma implies that
the canonical functor
$$
\begin{array}{ccc}
\Legos 	& \to 		& \Sh(\Legos,\tau)\\
U		&\mapsto	& [\underline{U}:=\Hom(.,U):V\mapsto \Hom(V,U)]
\end{array}
$$
is an embedding. A common denomination for a sheaf $F\in \Sh(\Legos,\tau)$ is
that of the functor of points of the corresponding space with values in $\Legos$.

\begin{center}
\bf This thus gives a definition of spaces by their points.
\end{center}

In usual (finite-dimensional) geometry, one
defines a particular class of spaces, called geometrical sheaves.
These spaces are covered by a special kind of morphism
$$\coprod U_{i}\to F$$
from a union of legos. The precise definition of such geometric contexts
can be found in \cite{Toen-Vaquie-CACA}. We do not give it in details because it
does not apply to (usually infinite-dimensional) spaces of maps,
which are the central objects of covariant field theory.

We now recall the general definitions from \cite{Kashiwara-Schapira-categories-and-sheaves}.

\begin{definition}
Let $\Legos$ be a category. A Grothendieck topology $\tau$ on $\Legos$ is the
datum of families of morphisms $\{f_{i}:U_{i}\to U\}$, called covering families,
and denoted $\Cov_{U}$,
such that the following holds:
\begin{enumerate}
\item (Identity) The identity map $\{\id:U\to U\}$ belongs to $\Cov_{U}$.
\item (Refinement) If $\{f_{i}:U_{i}\to U\}$ belongs to $\Cov_{U}$ and
$\{g_{i,j}:U_{i,j}\to U_{i}\}$ belong to $\Cov_{U_{i}}$, then the composed covering
family $\{f_{i}\circ g_{i,j}:U_{i,j}\to U\}$ belongs to $\Cov_{U}$.
\item(Base change) If $\{f_{i}:U_{i}\to U\}$ belongs to $\Cov_{U}$ and
$f:V\to U$ is a morphism, then $\{f_{i}\times_{U}f:U_{i}\times_{U}V\to V\}$ belongs
to $\Cov_{V}$.
\item (Local nature) If $\{f_{i}:U_{i}\to U\}$ belongs to $\Cov_{U}$ and
$\{f_{j}:V_{j}\to U\}$ is a small family of arbitrary morphisms such that
$f_{i}\times_{U} f_{j}:U_{i}\times_{U}V_{j}\to U_{i}$ belongs to $\Cov_{U_{i}}$,
then $\{f_{j}\}$ belongs to $\Cov_{U}$.
\end{enumerate}
A category $\Legos$ equipped with a Grothendieck topology $\tau$ is
called a site.
\end{definition}

We remark that for this definition to make sense one needs the fiber products that
appear in it to exist in the given category. A more flexible and general definition
in terms of sieves can be found in \cite{Kashiwara-Schapira-categories-and-sheaves}.

Suppose that we work on the category $\Legos=\Open_{X}$ of open subsets of a given
topological space $X$, with inclusion morphisms and its usual topology.
The base change axiom then says that a covering of
$U\subset X$ induces a covering of its open subsets. The local character means that families
of coverings of elements of a given covering induce a (refined) covering.

\begin{definition}
Let $(\Legos,\tau)$ be a category with Grothendieck topology.
A functor $X:\Legos^{op}\to \Sets$ is called a sheaf if the sequence
$$
\xymatrix{
X(U)\ar[r] & \prod_{i}X(U_{i})\ar@<0.7ex>[r]\ar@<-0.4ex>[r] & \prod_{i,j}X(U_{i}\times_{U}U_{j})
}
$$
is exact.
The category of sheaves is denoted $\Sh(\Legos,\tau)$.
\end{definition}

One can think of a sheaf in this sense as something analogous to the continuous
functions on a topological space. A continuous function on an open set is uniquely
defined by a family of continuous functions on the open subsets of a given covering,
whose values are equal on their intersections.

From now on, one further supposes that, for all legos $U$, $\underline{U}$ is a sheaf
for the given topology.

The main advantage of the category of sheaves is the following classical fact. 
\begin{lemma}
Let $\Open$ be the category of open subsets of $\R^n$ for varying $n$
with its usual topology. Let $U=U_{1}\coprod_{U}U_{2}$ be a covering of an open subset
by two open subsets.
Then one has
$$\underline{U}=\underline{U_{1}}\coprod_{\underline{U}}\underline{U_{2}}$$
if the coproduct is taken in the category of sheaves on $\Open$.
\end{lemma}

The most standard example of a punctual geometrical setting is given by the
theory of diffeology, which was developed by Smith \cite{Smith-1966} and
Chen \cite{Chen-1977}, and used in the physical setting by Souriau
(see \cite{Souriau} and \cite{piz} for references and historical background,
and \cite{Baez-diffeo} for an overview)
to explain the geometric methods used by physicists to study variational problems.

\begin{definition}
Let $\Open_{\Cc^\infty}$ be the category of open subsets of $\R^n$
for varying $n$ with smooth maps between them and $\tau$ be the
usual topology on open subsets. The category of diffeologies
(also called smooth spaces) is the category
$\Sh(\Open_{\Cc^\infty},\tau)$.
\end{definition}

There is a fully faithful embedding of the category of smooth varieties
in diffeologies that sends a variety $X$ to its functor of parameterized points
$$
\begin{array}{cccc}
\underline{X}: 	& \Open_{\Cc^\infty}^{op}  	& \to 		& \Sets\\
			& U						&\mapsto	& \Hom_{\Cc^\infty}(U,X).
\end{array}
$$

If we replace $\Open_{\Cc^\infty}$ by its opposite category $\Open_{\Cc^\infty}^{op}$,
we get the functional viewpoint of varieties:
there is a natural embedding of the category of smooth varieties in the category
of covariant functors from $\Open_{\Cc^\infty}$ to $\Sets$ given by sending
$X$ to $\Hom(X,.)$.
Since every open subset $U$ of $\R^n$ can be thought as given with an
embedding $U\subset \R^n$, we always have
$$\Hom(X,U)\subset \Hom(X,\R^n)=\Hom(X,\R)^n.$$
This simple result implies that from the functional viewpoint, it is often
enough to consider functions with values in $\R$, and this is actually
what analysts usually do, and they are right!

\begin{comment}
In the functional viewpoint of geometry, the cosheaf condition reads that
$$X(U)\to X(\Uc)$$
is a bijection between the space of functions with values in $U$
and their family of values on a full family of coordinates $\Uc$ on $U$.

The interesting property of the category of set valued contravariant functors is that
it contains all limits and colimits and has internal homomorphisms
(spaces of fields in the physicists language).
The sheaf condition implies that colimits of building blocks give reasonable spaces,
since it is a topological condition on pasting of building blocks (see \cite{Voevodsky2}
for a discussion of this in the algebraic context).
\end{comment}

We now come to the definition of contravariant constructions on spaces.
\begin{definition}
Let $(\Legos,\tau)$ be a category with Grothendieck topology and let
$X$ be a space (i.e., a sheaf for the given topology).
Let $C$ be a given category (of types of structure... think of vector spaces) and let
$\Omega:\Legos^{op}\to C$ be a contravariant functor (given by a standard contravariant
structure on $\Legos$... think of differential forms) that is moreover a sheaf.
If $\Legos/X$ denotes the
category of morphisms $U\to X$ for $U\in \Legos$, one defines $\Omega_{X}$
as the functor $\Omega_{X}:\Legos/X\to C$. An element of $\Omega_{X}$
is defined as a family of elements $\omega_{U}\in \Omega_{X}(U)$ compatible
with the functorial maps $\Omega_{U,V}:\Omega_{X}(U)\to \Omega_{X}(V)$
for $f:V\to U$ a morphism in $\Legos/X$.
\end{definition}

As an example, in the diffeological setting, the differential graded algebra of
differential form $U\mapsto (\Omega^*(U),d)$ on $\Legos$ has a well-defined
pull-back along smooth morphisms $f:U\to V$. This thus gives a definition
of the notion of differential form and a Rham complex on a diffeology.

\begin{definition}
Let $X:(\Open_{\Cc^\infty}^{op},\tau)\to \Sets$ be a diffeological space. A differential form
on $X$ is the datum, for every morphism $x:U\to X$ of a differential form $x^*\omega$,
such that for every $f:U\to V$ and $x':V\to X$ such that $x'\circ f=x$, one has
$$f^*((x')^*\omega)=x^*\omega.$$
\end{definition}

More generally, one can define the notion of a bundle over a diffeological space $X$
in a similar way by using the ``functor''
$$\Omega=\Bun:\Open_{\Cc^\infty}^{op}\to \Groupoids$$
that sends an open $U$ to the groupoid of bundles on $U$.
Let $X:(\Open_{\Cc^\infty}^{op},\tau)\to \Sets$ be a diffeological space. A bundle on $X$
is (roughly speaking) the datum, for every morphism $x:U\to X$ of a bundle $x^*E$ on $U$, and 
for every $f:U\to V$ and $x':V\to X$ such that $x'\circ f=x$, of an isomorphism
$$f^*((x')^*E)\cong x^*E,$$
with an additional identity associated to pairs of morphisms between points.
The proper mathematical formulation of this construction involves homotopical
methods, which we will talk about latter.

\begin{remark}
If we replace $\Legos$ by $\Legos^{op}$, and equip it with a Grothendieck topology $\tau'$
(also called a Grothendieck cotopology on $\Legos$), a $(\Legos^{op},\tau')$-space
will be defined by its spaces of functions $\Hom(X,.):\Legos\to \Sets$
that must fulfil a cosheaf property to preserve finite limits of $\Legos$. In this setting,
we get natural definitions of covariant operations like vector fields or differential
operators.
\end{remark}

%************************************************************************
\subsection{Equations and their solution functor with values in algebras}
\label{polynomial-solution}
Another take at the functorial approach to geometry is by the study
of equations and their solution functor (also called functor of points), as
explained in the introduction of the Springer version of EGA \cite{EGAI}.
Most of the spaces used in physics are described by some equations.

Let $\{P_{i}(x_{1},\dots,x_{n})\}$ be a family of polynomials with real coefficients.
To write down the equations $P_{i}=0$, one only needs a commutative
$\R$-algebra. If $A$ is a commutative $\R$-algebra, one can look for the solutions
of $P_{i}=0$ in $A^n$. This gives a functor
$$\Sol_{P_{i}=0}:\Alg_{\R}\to \Sets$$
defined by
$$\Sol_{P_{i}=0}(A):=\{(x_{1},\dots,x_{n})\in A^n,\forall i,\;P_{i}(x_{1},\dots,x_{n})=0\}.$$
The universal properties of the polynomial and of the quotient ring essentially tell
us that this functor is isomorphic to the functor
$$
\begin{array}{ccccc}
\uSpec(\R[x_{1},\dots,x_{n}]/(P_{i})):& \Alg_{\R} & \to & \Sets\\
& A & \mapsto & \Hom_{\Alg_{\R}}(\R[x_{1},\dots,x_{n}]/(P_{i}),A)
\end{array}
$$
corresponding to the algebra $\R[x_{1},\dots,x_{n}]/(P_{i})$.
The localization maps $A\to A[f^{-1}]$ for $f\in A$ define a Grothendieck
topology $\tau_{Zar}$ on $\Legos=\Alg_{\R}^{op}$ called the Zariski topology. The functor
$\Sol_{P=0}$ is a sheaf for this topology.

\begin{definition}
An algebraic space over $\R$ is a sheaf on the site $(\Alg_{\R}^{op},\tau_{Zar})$, i.e.,
a covariant functor
$$X:\Alg_{\R}\to \Sets$$
that is a sheaf fo $\tau_{Zar}$.
\end{definition}

A scheme over $\R$ (of finite type) is essentially an algebraic space
that ``can be covered'' (for more details, see
\cite{Toen-Vaquie-CACA}) by some solution spaces.

If one wants to work with equations defined by smooth functions, one has
to refine the $\Legos$ category to a particular kind of algebra.
For $A$ a real algebra and $a\in A$, we denote
$\Spec_{\R}(A):=\Hom_{\Alg_{\R}}(A,\R)$,
and
$$
\begin{array}{ccccc}
\ev_{a}:	&\Spec_{\R}(A)	& \to 		& \R\\
		& x							&\mapsto 	& x(a). 
\end{array}
$$
\begin{definition}
An algebra $A$ is called
\begin{enumerate}
\item smoothly affine if for every $n$, the natural map
$$
\begin{array}{ccc}
\uSpec(\Cc^\infty(\R^n))(A):=\Hom_{\Alg_{\R}}(\Cc^\infty(\R^n),A)&\to& A^n\\
\phi								&\mapsto & \phi(x_{1}),\dots,\phi(x_{n})
\end{array}
$$
is bijective;
\item smoothly closed geometric
if for every $a_{1},\dots,a_{n}\in A^n$ and $f\in \Cc^\infty(\R^n)$,
there exists a unique $a\in A$ such that the function
$$f\circ(\ev_{a_{1}}\times \dots\times\ev_{a_{n}}):\Spec_{\R}(A)\to \R$$
equals $\ev_{a}$.
\end{enumerate}
We denote $\Alg_{sa}$ (respectively, $\Alg_{scg}$) the category of smoothly affine (respectively,
smoothly closed geometric) real algebras.
\end{definition}
The notion of smoothly closed geometric algebra was first introduced by Nestruev \cite{Nestruev}.
\begin{comment}
\item geometric if the evaluation map
$$
\begin{array}{cccc}
\ev:	&A&\to		& \Hom(\Spec_{\R}(A),\R)\\
	& a&\mapsto & [x\mapsto x(a)]
\end{array}
$$
is injective.
\item smoothly closed if of every $a_{1},\dots,a_{n}\in A^n$ and $f\in \Cc^\infty(\R^n)$,
there exists an $a\in A$ such that the function
$$f\circ(\ev_{a_{1}}\times \dots\times\ev_{a_{n}}):\Spec_{\R}(A)\to \R$$
equals $\ev_{a}$.
\end{itemize}
We denote respectively $\Alg_{sa}$, $\Alg_{sc,\R}$ and $\Alg_{g,\R}$
the categories of smoothly affine, smoothly closed and smoothly geometric algebras.
We denote $\Alg_{scg,\R}$ the category of smoothly closed geometric algebras.
\end{definition}
\end{comment}

%The main advantage of the category $\Alg_{sa,\R}$ on $\Alg_{\R}$ is that
%one has
%$$\Cc^\infty(\R^2)\cong \Cc^\infty(\R)\otimes \Cc^\infty(\R)$$
%in it.
The main advantage of smoothly affine algebras over usual algebras is that they allow us
to make sense of the solution space to smooth equations.
If $\{f_{i}:\R^n\to \R\}$ is a family of smooth functions on $\R^n$,
one can look for the solution to $f_{i}(x_{1},\dots,x_{n})=0$ for
$(x_{1},\dots,x_{n})\in A^n$ if $A$ is in $\Alg_{sa}$. Indeed,
$(x_{1},\dots,x_{n})$ corresponds to a morphism
$\phi_{x_{1},\dots,x_{n}}:\Cc^\infty(\R^n)\to A$
and one can evaluate it at the $f_{i}$ to get elements $\phi_{x_{1},\dots,x_{n}}(f_{i})\in A$.
This gives a functor
$$
\begin{array}{ccccc}
\Sol_{f_{i}=0}:& \Alg_{sa}	 &\to&  \Sets\\
			 & A & \mapsto & \{(x_{1},\dots,x_{n})\in A^n,\;\phi_{x_{1},\dots,x_{n}}(f_{i})=0\}.
\end{array}
$$
It is then easy to show that one has a functor isomorphism
$$\Sol_{f_{i}=0}\cong \uSpec(\Cc^\infty(\R^n)/(f_{i})):\Alg_{sa}\to \Sets$$
(we carefully inform the reader that $\Cc^\infty(\R^n)/(f_{i})$ is not smoothly affine in general).

\begin{theorem}
\label{smoothly-affine-basic}
\begin{enumerate}
\item An algebra that is smoothly closed geometric is smoothly affine.
\item Let $\otimes_{sa}$ be the tensor product in $\Alg_{sa}$. Then
$$\Cc^\infty(\R^p)\otimes_{sa}\Cc^\infty(\R^q)\cong \Cc^\infty(\R^{p+q}).$$
\item If $U\subset \R^n$ is an open subset then $\Cc^\infty(U)$ is smoothly closed and geometric.
\item The functor
$$
\begin{array}{cccc}
\Cc^\infty: 	& \Open_{\Cc^\infty} & \to 		&	 \Alg_{sa}\\
			& U					&\mapsto	& \Cc^\infty(U)
\end{array}
$$ is fully faithful with essential image denoted $\Alg_{\Cc^\infty}$. It induces
an equivalence of sites between $(\Open_{\Cc^\infty},\tau)$ and $(\Alg_{\Cc^\infty}^{op},\tau_{Zar})$.
\item If $X$ is a smooth variety, then $\Cc^\infty(X)$ is smoothly closed and geometric.
\item If $\otimes_{scg}$ denotes the tensor product in $\Alg_{scg}$ and $M$ and $N$
are smooth varieties, then
$$\Cc^\infty(M)\otimes_{scg}\Cc^\infty(N)\cong \Cc^\infty(M\times N).$$
\item There is a natural adjoint $\Alg_{\R}\to \Alg_{scg}$
to the forgetful functor $\Alg_{scg}\to \Alg_{\R}$. It is called the smooth geometric closure, and is
denoted $A\mapsto \overline{A}^{scg}$.
\end{enumerate}
\end{theorem}
\begin{proof}
Let $A$ be a smoothly closed geometric algebra. Let $(a_{1},\dots,a_{n})\in A^n$
be a family of elements in $A^n$ and $f\in\Cc^\infty(\R^n)$ be a smooth function.
By definition, there exists a unique $a_{f}\in A$ such that
$$f\circ (\ev_{a_{1}}\times \dots\times\ev_{a_{n}})=\ev_{a}:\Spec_{\R}(A)\to \R.$$
The map
$$
\begin{array}{ccccc}
\phi_{a_{1},\dots,a_{n}}:	& \Cc^\infty(\R^n)&\to &A\\
						& f & \mapsto & a_{f}
\end{array}
$$
is an algebra morphism by construction, and it is uniquely defined by $a_{1},\dots,a_{n}$.
This gives the fact that the natural map
$$\Hom(\Cc^\infty(\R^n),A)\to A^n$$
is bijective, so that $A$ is smoothly affine.
By definition of the tensor product, if $A$ is smoothly affine, one has a bijection
$$
\begin{array}{ccc}
\Hom_{\Alg_{sa}}(\Cc^\infty(\R^p)\otimes_{sa}\Cc^\infty(\R^q),A) & \cong &
\Hom(\Cc^\infty(\R^p),A)\times \Hom(\Cc^\infty(\R^q),A)\\
&=&A^p\times A^q.
\end{array}
$$
But since $A$ is smoothly affine, one gets
$$A^{p+q}\cong \Hom_{\Alg_{sa}}(\Cc^\infty(\R^{p+q}),A).$$
This shows that $\Cc^\infty(\R^p)\otimes_{sa}\Cc^\infty(\R^q)$ and
$\Cc^\infty(\R^{p+q})$ have the same spectrum functor, so that they
are canonically isomorphic (they fulfil the same universal property).
%Suppose that $A$ is smoothly affine. A morphism $f:\Cc^\infty(\R^n)\to A[\epsilon]/(\epsilon^2)$
%can be written as a sum $f_{0}+\epsilon f_{1}$ with $f_{0}:\Cc^\infty(\R^n)\to A$ given
%by a elements $(a_{1},\dots,a_{n})\in A^n$ and $D=f_{1}:\Cc^\infty(\R^n)\to A$ a $\Cc^\infty(\R^n)$-derivation.
The other statements are translations of the results in \cite{Nestruev}.
\end{proof}

The above proposition tells us that one can translate statements about
diffeological spaces purely in algebraic terms, if one works with the category
$\Alg_{\Cc^\infty}$ or $\Alg_{scg}$. Moreover, in this setting, one can talk of the solution space
for a given smooth equation.

\begin{comment}
We will see later on that the natural setting to talk about the functor of
solutions to a given polynomial partial differential equation
$$F(t,\partial_{t}^ix)=0$$
on the sections of a bundle $\pi:C\to M$ is given by the setting of
$\Dc_{M}$-algebras.
\end{comment}

\begin{definition}
Consider the category $\Alg_{scg}^{op}$, equipped with the Zariski topology
$\tau_{Zar}$ (generated by the localization maps $A\to A[f^{-1}]$).
A smoothly algebraic space is a sheaf $X$ on the site $(\Alg_{scg}^{op},\tau_{Zar})$.
\end{definition}

\begin{comment}
The advantage of smoothly algebraic spaces on diffeological spaces is that
they contain infinitesimal spaces of the form $\uSpec(\R[\epsilon]/(\epsilon^2))$.
This allows one to define a nice notion of vector field on a given smoothly algebraic
space, which is impossible for diffeological spaces. Moreover, they also allow
us to talk about the space of solutions of an equation given by smooth functions,
which is impossible if we work with usual $\R$-algebras.
\end{comment}

\begin{definition}
Let $X:\Alg_{\R}\to \Sets$ be an algebraic space. The tangent
bundle of $X$ is defined by
$$\pi:TX(A):=X(A[\epsilon]/(\epsilon^2))\to X(A),$$
where $\pi$ is induced by the projection $A[\epsilon]/(\epsilon^2)\to A$ that
sends $\epsilon$ to $0$.
A vector field on $X$ is a section $\vec v:X\to TX$ of $\pi:TX\to X$.
\end{definition}

\begin{proposition}
Let $X$ be a usual smooth variety, seen as an algebraic space $X_{a}$.
Then a vector field on $X_{a}$ identifies with a usual vector field on $X$.
\end{proposition}
\begin{proof}
Recall that vector fields on $X$ identify with derivations of $\Cc^\infty(X)$.
Let $\frac{\partial}{\partial \vec v}$ be such a derivation. Then if $A$ is smoothly affine,
the map
$$\vec v:X_{a}(A)\to TX_{a}(A)$$
that sends $f:\Cc^\infty(X)\to A$ to
$$
\vec v(f)=f+\epsilon.f\circ \frac{\partial}{\partial \vec v}:\Cc^\infty(X)\to A[\epsilon]/(\epsilon^2)
$$
is well-defined because its image is additive and
$$
\begin{array}{ccc}
f(gh)+\epsilon.f(\frac{\partial(gh)}{\partial \vec v}) & = &
f(g).f(h)+\epsilon.f(\frac{\partial g}{\partial \vec v}.h+g.\frac{\partial h}{\partial \vec v})\\
&=&f(g).f(h)+\epsilon. f(\frac{\partial g}{\partial \vec v}).f(h)+f(g).f(\frac{\partial h}{\partial \vec v})\\
&=&[f(g)+\epsilon.f(\frac{\partial g}{\partial \vec v})].[f(h)+\epsilon.f(\frac{\partial h}{\partial \vec v})].
\end{array}
$$
Conversely, given a vector field $\vec v:X_{a}\to TX_{a}$, and since
$\Cc^\infty(X)$ is smoothly affine by theorem \ref{smoothly-affine-basic},
one can compute the value of $\vec v$ on the universal point
$$\id_{X}\in X_{a}(\Cc^\infty(X)).$$
This gives a morphism
$$\vec v(\id_{X})=\id_{X}+\epsilon.D:\Cc^\infty(X)\to \Cc^\infty(X)[\epsilon]/(\epsilon^2)$$
where $D:\Cc^\infty(X)\to \Cc^\infty(X)$ is a derivation.
We have thus shown the equivalence between the two notions in this case.
\end{proof}

%************************************************************************
\subsection{A simple example of differential calculus on spaces}
\label{simple-example}
We will now give a simple but generic example of differential calculus on spaces
of fields. The same computational method works whenever the category $\Legos$ has
a well-behaved notion of differential form (for example, for algebraic spaces or superspaces).

Consider the lagrangian variational problem of Newtonian
mechanics, with fiber bundle
$$\pi:C=\R^3\times [0,1]\to [0,1]$$
whose sections are smooth maps $x:[0,1]\to \R^3$, which represent
a material point moving in $\R^3$. The space of histories is
given by fixing a pair of starting and ending points for trajectories
$\{x_{0},x_{1}\}$, i.e.,
$$H=\{x\in \uGamma(M,C),\;x(0)=x_{0},\;x(1)=x_{1}\}.$$
More precisely, if $U$ is a parameterizing open subset of some $\R^n$,
$$
H(U)=
\{x(t,u)\in \uGamma(M,C)(U)\cong \Hom([0,1]\times U,\R^3),
\;x(0,u)=x_{0},\;x(1,u)=x_{1}\}.
$$
If $\langle.,\rangle$ is the standard metric on $\R^3$ and $V:\R^3\to \R$ is
a given ``potential'' function, one defines
\footnote{We remark that for this
to be well-defined, one needs to put a domination condition on $x(t,u)$ to apply Lebesgue's
dominated derivation under the integral. We will implicitely add this condition
everywhere in this paper.}
the action functional morphism
$S:H\to \R$ by
$$
\begin{array}{ccccc}
S_{U}: 	& H(U) & \to & \Cc^\infty(U)=\R(U)\\
		& x(t,u) & \mapsto &\int_{M}\frac{1}{2}m\|\partial_{t}x(t,u)\|^2+V(x(t,u))dt
\end{array}
$$
for $U$ a parameterizing open subset of some $\R^n$.
The differential of $S$ is a one form $\omega$ on $\uGamma(M,C)$, which is formalized
mathematically as a family of one forms $\{x^*\omega\}$ on $U$ for
every morphism $x:U\to \uGamma(M,C)$, which are compatible with pull-back
along commutative diagrams
$$
\xymatrix{
U\ar[r]^x & \uGamma(M,C)\\
V\ar[ur]_{x'}\ar[u]^{f} &}
$$
meaning that $f^*x^*\omega=(x')^*\omega$.
Concretely, if $x:M\times U\to C$ is given and $\vec u$ is a vector field on $U$, one has
$$
\langle d_{x}S_{U},\vec u\rangle=
\frac{\partial}{\partial \vec u}\left[\int_{M}\frac{1}{2}m\|\partial_{t}x(t,u)\|^2-V(x(t,u))dt\right],
$$
so that
$$
\langle d_{x}S_{U},\vec u\rangle=
\int_{M}\frac{1}{2}m\frac{\partial}{\partial \vec u}\|\partial_{t}x\|^2-\frac{\partial}{\partial \vec u}V(x)dt
$$
and one gets
$$
\langle d_{x}S_{U},\vec u\rangle=
\int_{M}m\Langle\partial_{t}x,\frac{\partial}{\partial \vec u}\partial_{t}x\Rangle-
\Langle d_{x}V(x),\frac{\partial x}{\partial \vec u}\Rangle dt.
$$
By permuting the $\vec u$ and $t$ derivative and integrating by parts in $t$, one gets
$$
\langle d_{x}S_{U},\vec u\rangle=
\int_{M}\Langle-m\partial_{t}^2x-d_{x}V(x),\frac{\partial x}{\partial \vec u}
\Rangle dt+m\left[\Langle\partial_{t}x,\frac{\partial x}{\partial \vec u}\Rangle\right]_0^1.
$$
Since $x(0,u)=x_{0}$ and $x(1,x)=x_{1}$ are constant in $u$,
the boundary term vanishes and one gets finally
$$
\langle d_{x}S_{U},\vec u\rangle=
\int_{M}\Langle-m\partial_{t}^2x-d_{x}V(x),\frac{\partial x}{\partial \vec u}
\Rangle dt.
$$
The condition that $d_{x}S=0$ is then equivalent to the usual Newton equation,
so that the $U$-valued points of the space of trajectories are
$$T(U)=\{x\in H(U)|d_{x}S_{U}=0\}=\{x\in H(U)|m\partial_{t}^2x(t,u)=-V'(x(t,u))\}.$$

We remark that the above computation is completely standard in physics,
and we just gave a mathematical language to formulate it. Usually,
one uses functional analytic methods here, but they do not generalize properly
to the super (fermionic) case, contrary to ours. Moreover, the main input of
this mathematical formulation is that the spaces of histories $H$ and
trajectories $T$ are exactly of the same nature as the spaces of parameters
$M$ and of configurations $C$.

\begin{comment}
If we work with smoothly algebraic spaces, meaning
sheaves on the site $(\Alg_{\R}^{sa},\tau_{Zar})$, we can also
make this computation using vector fields on the space of histories $H$.
First, if $A$ is smoothly affine, $A[\epsilon]/(\epsilon^2)$ also, and
we thus have a family of natural isomorphisms
$$
\begin{array}{ccc}
T\R(A) 	& = 		& \Hom(\Cc^\infty(\R),A[\epsilon]/(\epsilon^2))\\
		&\cong 	& A[\epsilon]/(\epsilon^2)\\
		&\cong 	& A^2\\
		&\cong  & \Hom(\Cc^\infty(\R^2),A)\\
		&\cong  & \R^2(A)
\end{array}
$$
So that $T\R\cong \R^2$.
\end{comment}

%************************************************************************
\subsection{Various types of spaces used in physics}
The main objects of classical field theory are spaces of functions
$$\uHom(X,Y)$$
between two given spaces $X$ and $Y$.
As explained in the previous sections, one can see these spaces
as spaces similar to $X$ and $Y$, if one embeds all of them in a
category of sheaves on a Grothendieck site $(\Legos,\tau)$.
The choice of the site will then depend on the needs of the situation:
if one needs only differential forms and $X$ and $Y$ are usual smooth varieties,
the diffeological setting will be sufficient and
$\Legos=\Open_{\Cc^\infty}$.
However, to have a notion of vector field on spaces, one will need a category
of $\Legos$ given by some algebras, e.g.
$\Legos=\Alg_{\R}^{op}$
for algebraic spaces. If one wants to work with solution spaces to smooth equations,
one can also use $\Legos=\Alg_{scg}^{op}$.

To describe supervarieties and spaces of morphisms between them, one can
not really avoid working with legos given by categories $\Legos=\Alg_{s,\R}^{op}$
opposite to superalgebras (see the forthcoming sections).
This is a good reason to work from the start with algebras also in the classical smooth case.

We will also see that similar functorial and scheme theoretic methods can be applied to
local functional calculus, and to the study of solution spaces to
nonlinear partial differential equations. In this case, one will have
$\Legos=\Alg_{\Dc}^{op}$, the category of $\Dc$-algebras (see the forthcoming section).

Since the resolution of an equation (or more generally of a problem)
is sometimes obstructed, one will also work with more general types of space,
of a homotopical nature, that encode the obstructions to the resolution of
the given equations. We recall from Toen and Vezzosi's work \cite{Toen-HAG-Essen}
the idea of the construction of these spaces that allow a geometric treatment of obstruction theory.

These are essentially, in their most general form, given by some homotopy classes of functors
$$X:(\Legos,\tau,W)^{op}\to (C,W_{C})$$
where $(\Legos,W)$ is a model category (category with a notion of weak equivalences)
equipped with the homotopical analog of
a Grothendieck topology $\tau$ and $(C,W_{C})$ is a model category. The definition of
the homotopy equivalence relation on these functors involves not only the week equivalences
$W$ and $W_{C}$ in $\Legos$ and $C$ but also the topology $\tau$.
We do not go into the details of their quite technical definition, but illustrate it by physical examples.

Let $\Delta$ be the category of ordered sets of the form $[1,\dots,n]$ with
increasing morphisms. If $C$ is a category, we denote $C_{\Delta}$ the
category of functors $\Delta^{op}\to C$ and call it the category of simplicial
objects in $C$. The first homotopical generalization of spaces one can do
is to consider as building blocks a usual category $\Legos$ (for example the category
$\Open_{\Cc^\infty}$ of building blocks for diffeologies)
and to consider simplicial presheaves
$$X:\Legos^{op}\to \Sets_{\Delta}=:\SSets$$
as spaces.
This gives a version of the theory of stacks, which are necessary
for studying parameter spaces for objects up to isomorphisms.
The main physical example of the usefulness of this setting is gauge theory:
if $G$ is a group, the space of principal $G$-bundles on $M$ is to be considered
as (the homotopy class of) a simplicial presheaf
$$\Bun_{G}(M):\Open_{\Cc^\infty}^{op}\to \SSets$$
whose coarse moduli space is simply the set of isomorphism classes of 
principal $G$-bundles on $M$ parameterized by $U$.
There is no physical reason for choosing one particular $G$-bundle
on $M$ and this explains why it is natural to work directly with the universal
principal bundle $E=P\to\Bun_{G}(M)$ in gauge field theory.

The obstruction theory to infinitesimal deformation of stacks can not
be dealt with properly without sinking a little bit further in homotopical methods. 
Derived geometry is a homotopical generalization of spaces which uses as building
blocks simplicial categories $\Legos_{\Delta}$ and as spaces simplicial presheaves
$$X:\Legos_{\Delta}^{op}\to \Sets_\Delta=:\SSets.$$
These are useful for studying derived moduli spaces and necessary
for making deformation and obstruction theory (for example the cotangent complex)
functorial. For example, if one starts from usual scheme theory where $\Legos$
is the category of algebras, one gets as homotopical counterpart a geometry where
building blocks are given by the differential graded category of differential graded
(or simplicial) algebras.
These kinds of space naturally appear in the quantization of gauge theory,
as one can guess from the BRST-BV method for quantization
(see \cite{Henneaux-Teitelboim} and the articles of Stasheff et al., for example
\cite{fulp-lada-stasheff}).

We refer to Toen and Vezzosi's long opus \cite{Toen-Vezzosi-HAG-II}
for theory and mathematical applications of these homotopical methods.
We will talk a bit along the way about their physical applications.

In all these context, given two spaces $X$ and $Y$ constructed from
a given category of building blocks $\Legos$, one can easily define
the corresponding mapping space $\underline{\Hom}(X,Y)$ as the sheaf on $\Legos$
associated to the presheaf of sets
$$U\mapsto \Hom_{\Legos/U}(X\times U,Y\times U).$$
The homotopical settings need more care (resolutions) but are essentially similar.

In lagrangian relativistic physics, the space of (bosonic) Feynman history is simply
the space of all sections $s:M\to E$ of a given bundle $E$
on spacetime $M$. Since they are infinite dimensional, such functional
spaces are not so easily studied using the observable viewpoint of
physics (i.e. functions on them) because it is hard to find what
is a natural and general notion of function (i.e. of an observable) on
such a space. We remark also that the bundle $E$ need not be linear
(interaction particles are given by connections, which form an affine
bundle)
and can also be a moduli stack (as for example the moduli stack of principal
bundles, which is the natural setting for gauge theory).

The punctual viewpoint is much closer to the physicists' viewpoint
and allows us to consider completely canonical geometrical structures on functional
spaces.

We now turn to an abstract version of geometry that allows us to
treat classical spaces (bosonic) and superspaces (fermionic) with a unique
and concise language.

%************************************************************************
\subsection{Generalized algebras and relative geometry}
\label{relative-geometry}
So far, we have been working with the $\Legos$ categories $\Alg_{\R}^{op}$ and
$\Alg_{scg}^{op}$ of usual and smoothly closed geometric algebras. The spaces used
by physicists in variational calculus (for example superalgebras or $\Dc$-algebras)
are based on a generalized notion of algebra,
which is usually defined as a monoid in a given symmetric monoidal
(sometimes model) category.

We recall the definition of a scheme relative to a given symmetric monoidal
category, due to Toen and Vaqui\'e \cite{Toen-Vaquie}. This work is of course
very much inspired by Grothendieck's viewpoint of geometry. This approach
to supergeometry was also emphasized in Deligne's lecture notes
\cite{notes-on-supersymmetry}.
These are spaces whose building blocks are commutative monoids in the given monoidal category. They can also be generalized to the homotopical situation of symmetric
monoidal model categories, as explained by Toen and Vezzosi
\cite{Toen-Vezzosi-HAG-II}. We will need these generalizations but prefer
to restrict ourselves to the classical case since it is sufficiently instructive.

These methods are necessary for systematically studying spaces of fermionic-valued
fields $\psi:M\to \Pi S$ from the punctual viewpoint, like for example electrons
on usual spacetime, because these spaces are superspaces.
They also cannot be avoided if one wants to study superalgebras geometrically
(for example if $M$ is a supermanifold, as in supersymmetric field theories), and the
use of monoidal categories greatly simplifies the computations since it
allows one to completely forget about signs and to work with superspaces
as if these were usual spaces.

Let $K$ be a base field of characteristic $0$.
\begin{definition}
A symmetric monoidal category over $K$ is a tuple $(\Cc,\otimes,\1,\un,\as,\com)$ composed of
\begin{enumerate}
\item an abelian $K$-linear category $\Cc$,
\item a $K$-linear bifunctor $\otimes:\Cc\times \Cc\to \Cc$,
\item an object $\1$ of $\Cc$ called the unit object,
\item for each object $A$ of $\Cc$, two unity isomorphisms $\un^r_{A}:A\otimes \1\to A$
and $\un^l_{A}:\1\otimes A\to A$,
\item for each triple $(A,B,C)$ of objects of $\Cc$, an associativity isomorphism
$$\as_{A,B,C}:A\otimes (B\otimes C)\to (A\otimes B)\otimes C,$$
\item for each pair $(A,B)$ of objects of $\Cc$, a commutativity isomorphism
$$\com_{A,B}:A\otimes B\to B\otimes A,$$
\end{enumerate}
that are supposed to fulfil (for more details, we refer the reader to the article on monoidal
categories on wikipedia)
\begin{enumerate}
\item a pentagonal axiom for associativity isomorphisms,
\item a compatibility of unity and associativity isomorphisms,
\item an hexagonal axiom for compatibility between the commutativity and
the associativity isomorphisms, and
\item the idempotency of the commutativity isomorphism:
$\com_{A,B}\circ \com_{B,A}=\id_{A}$.
\end{enumerate}
The tensor category is called closed if it has internal homomorphisms, i.e.,
if for every pair $(B,C)$ of objects of $\Cc$, the functor
$$A\mapsto \Hom(A\otimes B,C)$$
is representable by an object
$\uHom(B,C)$ of $\Cc$.
\end{definition}

The main example of a closed commutative tensor category is the category
$\Vect_{K}$ of $K$-vector spaces. The idea for defining differential calculus
on algebras in an abstract symmetric monoidal category is to formalize it
for usual algebras using only the tensor structure and morphisms in $\Vect_{K}$.

Consider now the category whose objects are graded vector spaces
$$V=\oplus_{k\in \Z}V^k$$
and whose morphisms are linear maps respecting the grading.
We denote it $\Vect_{g}$. A graded vector space restricted to
degree $0$ and $1$ is called a supervector space, and we denote
$\Vect_s$ the category of supervector spaces.
These are abelian and even $K$-linear categories.
If $a\in V^k$ is a homogeneous element of a graded vector space $V$,
we denote $\deg(a):=k$ its degree.
The tensor product of two graded vector spaces $V$ and $W$
is the usual tensor product of the underlying vector spaces equipped
with the grading
$$(V\otimes W)_{k}=\oplus_{i+j=k}V_{i}\otimes W_{j}.$$
There is a natural homomorphism object in $\Vect_{g}$, defined
by
$$\uHom(V,W)=\oplus_{n\in \Z}\Hom^n(V,W)$$
where the degree $n$ component $\Hom_{n}(V,W)$ is the set of all linear maps
$f:V\to W$ such that $f(V^k)\subset W^{k+n}$.
It is an internal homomorphism object meaning that for every $X$,
there is a natural bijection
$$\Hom(X,\uHom(V,W))\cong \Hom(X\otimes V,W).$$
The tensor product of two internal homomorphisms $f:V\to W$ and $f':V\to W'$
is defined using the Koszul sign rule on homogeneous components. We have
$$(f\otimes g)(v\otimes w)=(-1)^{\deg(g)\deg(v)}f(v)\otimes g(w).$$

The tensor product is associative with unit $\mathbf{1}=K$ in degree $0$,
the usual associativity isomorphisms of $K$-vector spaces.
The main difference with the tensor category $(\Vect,\otimes)$ of usual
vector space is given by the non-trivial commutativity isomorphisms
$$c_{V,W}:V\otimes W\to W\otimes V$$
defined by extending by linearity the rule
$$v\otimes w\mapsto (-1)^{\deg(v)\deg(w)}w\otimes v.$$
One thus obtains a symmetric monoidal category $(\Vect_g,\otimes)$
which is moreover closed, i.e., has internal homomorphisms.

Let $A$ be an associative unital ring.
Let $(\Mod_{dg}(A),\otimes)$ be the monoidal category of
graded (left) $A$-modules, equipped with a linear map $d:C\to C$ of degree $-1$
such that $d^2=0$, with graded morphisms that commute with $d$
and tensor product $V\otimes W$ of graded vector spaces, endowed with the
differential $d:d_{V}\otimes \id_{W}+\id_{V}\otimes d_{W}$
(tensor product of graded maps, i.e., with a graded Leibniz rule),
and also the same anticommutative commutativity constraint.
This is also a closed symmetric monoidal category.

\begin{definition}
Let $(\Cc,\otimes)$ be a symmetric monoidal category over $K$.
An algebra in $\Cc$ is a triple $(A,\mu,\nu)$ composed of
\begin{enumerate}
\item an object $A$ of $\Cc$,
\item a multiplication morphism $\mu:A\otimes A\to A$, and
\item a unit morphism $\nu:\1\to A$,
\end{enumerate}
such that for each object $V$ of $\Cc$, the above maps fulfil
the usual associativity, commutativity and unit axiom with
respect to the given associativity, commutativity and unity isomorphisms
in $\Cc$. We denote $\Alg_{\Cc}$ the category of algebras in $\Cc$.
\end{definition}

In particular, a superalgebra is an algebra
in the monoidal category $(\Vect_s,\otimes)$.
Recall that the commutativity of a superalgebra
uses the commutativity isomorphism of the tensor category $\Vect_s$,
so that it actually means a graded commutativity: $(A,\mu)$ is commutative
if $\mu\circ \com_{A,A}=\mu$.
We denote $\Alg_{s}$ the category of real superalgebras.

We now define, following Toen and Vaqui\'e \cite{Toen-Vaquie}, the notion of scheme
on $\Cc$. Given an algebra $A$ in $\Cc$, one defines the corresponding
affine scheme by its ``functor of points''
$$
\begin{array}{cccc}
\uSpec(A): 	& \Alg_{\Cc} 	& \to 	& \Sets\\
		& C			&\mapsto&\uSpec(A)(C):=\Hom(A,C)
\end{array}
$$
and Yoneda's lemma implies that there is a natural bijection
$$\Hom(A,B)\overset{\sim}{\to} \Hom(\uSpec(B),\uSpec(A))$$
between algebra morphisms and affine schemes morphisms.
\begin{definition}

An algebra morphism $f:A\to B$ (or the corresponding morphism of functors
$\uSpec(B) \to \uSpec(A)$) is called a standard Zariski open if it is
a flat and finitely presented monomorphism:
\begin{enumerate}
\item (monomorphism) for every algebra $C$, $\uSpec(B)(C)\subset \uSpec(A)(C)$,
\item (flat) the base change functor $\_\otimes_{A}B:A-\Mod\to B-\Mod$ is exact, and
\item (finitely presented) if $\uSpec(B/A)$ denotes $\uSpec(B)$ restricted to $A$-algebras,
$\uSpec(B/A)$ commutes with filtered inductive limits. 
\end{enumerate}
A family of morphisms $\{f_{i}:A\to A_{i}\}_{i\in I}$ is called a Zariski covering if
\begin{enumerate}
\item for each $i$, $A\to A_{i}$ is flat,
\item there exists a finite subset $J\subset I$ such that the functor
$$\prod_{i\in J}\_\otimes_{A}A_{j}:A-\Mod\to \prod_{j\in J}A_{j}-\Mod$$
is conservative, and
\item every $f_{i}:A\to A_{i}$ is Zariski open.
\end{enumerate}
The Grothendieck topology generated by the Zariski coverings is called the Zariski
topology. If $\uSpec(A):\Alg\to \Sets$ is an affine scheme and $U\subset X$
is a subfunctor, $U$ is called a Zariski open if it is the image of a morphism
$$\coprod_{i\in I}\uSpec(A_{i})\to \uSpec(A)$$
induced by standard Zariski open subschemes $\uSpec(A_{i})\subset \uSpec(A)$. 
If $X:\Alg_{\Cc}\to \Sets$ is any functor, a Zariski open in $X$ is
a subfunctor $U\subset X$ such that for all $\uSpec(A)\to X$,
$$U\times_{X}\uSpec(A)\subset \uSpec(A)$$
is a Zariski open morphism.
We say that $X$ is a (relative) scheme on $\Cc$ if
\begin{enumerate}
\item it is a sheaf for the Zariski topology, and
\item it has a covering by Zariski open subfunctors.
\end{enumerate}
\end{definition}

Using this definition, it is not hard to generalize
the basic functorial results on schemes in EGA to relative schemes on $\Cc$.
This also allows us to define group schemes in $\Cc$ in a completely transparent way.

%************************************************************************
\subsection{Relative differential calculus}
\label{diff-calc-sym-monoidal}
We essentially follow Lychagin \cite{Lychagin93} and also Krasilsh'chik and Verbovetsky
\cite{Krasilshchik-Verbovetsky-1998} here.
The generalization to the homotopical setting is done in
Toen and Vezzosi \cite{Toen-Vezzosi-HAG-II} for derivations.
%, but not
%for differential operators. We explain along the way the necessary
%modifications.

From now on, let $\Cc$ be a symmetric monoidal category over $K$ and $(A,\mu)$
be an algebra in $\Cc$. We will now define differential invariants of $(A,\mu)$.

A left $A$-module is an object $M$ of $\Cc$ equipped with an external
multiplication map $\mu^l_{M}:A\otimes M\to M$. If $A$ is a commutative algebra
in $\Cc$, one can put on $M$ a right $A$-module structure
$\mu^r:M\otimes A\to M$ defined by $\mu^r_{M}:=\mu\circ \com_{M,A}$.
We will implicitly use this right $A$-module structure in the
formulas below.

\begin{comment}
\begin{definition}
Let $M$ be an $A$-module and $\Sym^*_{A}(M)\to A$ be the augmentation morphism.
The Berezinian of $M$ is defined as
$$\Ber(M):=\uExt_{\Sym^*_{A}(M)}(A,\Sym^*_{A}(M^*)).$$
\end{definition}
If $M$ is a free module of rank $(p,q)$ on a commutative superalgebra $A$,
the only non-zero $\uExt$ is $\uExt^p$ and it is a free $A$-module of
rank $(1,0)$ for $q$ even and $(0,1)$ for $q$ odd.
\end{comment}

\begin{definition}
Let $M$ be an $A$-module. A morphism $D:A\to M$ in $\Cc$
is called a derivation if
$$D(fg)=fD(g)+D(f)g$$
which means more precisely that the morphism
$D\circ \mu:A\otimes A\to M$ is equal to the sum
$\mu^l_{M}\circ (\id_{A}\otimes D)+\mu^r_{M}\circ (D\otimes \id_{A})$.
\end{definition}

The description of $\Der(A,M)$ shows that it can be seen as a subobject $\uDer(A,M)$ of
the internal homomorphisms $\uHom(A,M)$ in $\Cc$ defined
by the kernel of the morphism
$\uHom(A,M)\to \uHom(A\otimes A,M)$ defined by
$$D\mapsto D\circ \mu-\mu^l_{M}\circ (\id_{A}\otimes D)+\mu^r_{M}\circ (D\otimes \id_{A}).$$
We remark that for $\Cc=\Vect_s$ this expression can be expressed as a graded Leibniz rule
by definition of the right $A$-module structure on $M$ above.

% We denote $\Der(A,M)$ the object of $\Cc$ given by all these derivations.
% Since $A$ is commutative, there is a natural $A$-module structure on $\Der(A,M)$
% given by $\mu:A\otimes \Der(A,M)\to \Der(A,M)$. Indeed, we have
% $$(hD)(fg)=hfD(g)+hD(f)g$$

The representing object for the functor $\Der:A-\Mod(\Cc)\to \Cc$ is the $A$-module
$\Omega^1_{A}$ of (Kaehler) $1$-forms.
One can restrict the derivation functor on $A$-modules to a subcategory $\Cc_{g}$ of $\Cc$,
and this allows us to define various other types of differential form on $A$,
called admissible differential forms for $\Cc_{g}$.
If $A$ is smoothly closed geometric, one usually uses geometric modules, since they
give back usual differential forms in the case $A=\Cc^\infty(X)$ for $X$ a manifold.
We recall their definition from \cite{Nestruev}.
\begin{definition}
Let $A$ be an $\R$-algebra. An $A$-module $P$ is called geometric
if
$$\cap_{x\in \Spec_{\R}(A)}\mfk_{x}M=0,$$
where $\mfk_{x}$ denotes the ideal of functions that annihilate at $x$,
i.e., the kernel of the map $x:A\to \R$.
We denote $\Mod_{g,A}$ the category of geometric modules.
\end{definition}

\begin{theorem}
Let $X$ be a smooth variety. Admissible differential forms for
the category $\Mod_{g,A}$ identify with usual differential forms
on $X$. In particular, if $U\subset \R^n$ is a lego, there is
an identification
$$\Omega^1(U)\cong \Gamma(U,T^*U)$$
where $T^*U:=U\times (\R^n)^*$ is the cotangent space on $U$.
\end{theorem}
\begin{proof}
See Nestruev \cite{Nestruev}, theorem 1.43.
\end{proof}

We now define the differential operators on $A$ following Lychagin in \cite{Lychagin93}.
Let $M$ and $N$ be two left $A$-modules. The internal homomorphisms
$\uHom(M,N)$ are naturally equipped with two $A$-module structures
$$
\mu^l:A\otimes \uHom(M,N)\to \uHom(M,N)\textrm{ and }
\mu^r:\uHom(M,N)\otimes A\to \uHom(M,N).
$$
Define the morphism $\delta^l:A\otimes\uHom(M,N)\to \uHom(M,N)$ by
$$\delta^l=\mu^{l}-\mu^{r}\circ \com_{\uHom(M,N),A}$$
and define $\delta^r:\uHom(M,N)\otimes A\to \uHom(M,N)$ by
$\delta^r=-\delta^l\circ \com_{A,\uHom(M,N)}$.
Let $\delta^l_{n}:A^{\otimes n}\otimes \uHom(M,N)\to \uHom(M,N)$
be the $n$-tuple composition of $\delta^l$.

\begin{definition}
The module of differential operators of order $n$ from $M$ to $N$ of order $n$ is the
subobject $\uDiff_{n}(M,N)$ of $\uHom(M,N)$ given by intersecting
the kernel of $\delta^l_{n}$ with $\uHom(M,N)$ in $A^{\otimes n}\otimes \uHom(M,N)$.
If $M$ is a fixed object of $A-\Mod(\Cc)$, the representing object for the functor
on $\uDiff_{n}(M,.):A-\Mod(\Cc)\to \Cc$ is called the jet module of $M$, and is denoted
$J^n(M)$.
\end{definition}

In the monoidal category of usual vector spaces, this gives back the usual definition
of differential operators and algebraic jet modules. To get smooth jet modules,
one has to work with the subcategory of geometric modules on smoothly closed
geometric algebras.

\begin{definition}
Let $A$ be a superalgebra. The heart $|A|$ of $A$ is the quotient of
$A$ by the ideal generated by the odd part $A^1$.
One calls $A$ smoothly closed geometric if its heart $|A|$ is smoothly closed geometric
and its odd part $A^1$ is a geometric module over $|A|$.
The corresponding category is denoted $\Alg_{s,scg}$.
\end{definition}

\begin{definition}
The affine superspace of dimension $n|m$ with values in $A$ is
the superspace with values on a superalgebra given by
$$\A^{n|m}(A)=(A^0)^n\oplus (A^1)^m.$$
\end{definition}

This superspace is affine. More precisely, one has
$$\A^{n|m}=\uSpec(\R[x_{1},\dots,x_{n};\theta_{1},\dots,\theta_{m}])$$
with $x_{i}$ commuting variables and $\theta_{i}$ anticommuting variables
in the algebraic setting and
$$\A^{n|m}=\uSpec(\Cc^\infty(\R^n)[\theta_{1},\dots,\theta_{m}])$$
in the smoothly closed geometric case.

The superspaces mostly used by physicists are spaces modeled on the category
$\Alg_{s,scg}$.

\begin{comment}
We now finish by the definition of the notion of connection.

\begin{definition}
Let $M$ be an $A$-module. A connection on $M$ is a morphism
$$\nabla:M\to \Omega^1_{A}\otimes_{A} M$$
such that
$$\nabla(f.m)=f.\nabla(m)+df\otimes m,$$
i.e., such that
$$
\nabla\circ \mu^l_{M}=\mu^l_{\Omega^1_{A}\otimes_{A} M}\circ (\id_{A}\otimes \nabla)+
d\otimes \id_{M}.
$$
\end{definition}

We carefully advice the reader that the left multiplication
$\mu^l_{\Omega^1_{A}\otimes_{A} M}$ by $A$ is twisted
by the commutativity isomorphism because the tensor product
uses the right $A$-module structure on $\Omega^1_{A}$.
This implies the consistency of our definition with the one
given by Verbovetsky in \cite{Verbovetsky} for the super case.

Let now $(\Cc,\Cc^0,F,C,W,\otimes)$ be a symmetric monoidal model category
fulfilling the axioms of \cite{Toen-Vezzosi-HAG-II} and $(A,\mu:A\otimes A\to A)$ be
a monoid in $\Cc$, i.e., an algebra relative to $(\Cc,\otimes)$.
We then define the module of homotopical differential operators between two
$A$-modules $M$ and $N$ as the homotopy fibered product
$$
\R\uDiff_{n}(M,N)=
\R\uHom(M,N)
\underset{A^{\overset{\Lb}{\otimes}n}\overset{\Lb}{\otimes}\R\uHom(M,N)}
{\overset{h}{\times}}
\Ker(\R\delta_{n}^l).
$$
\end{comment}

%************************************************************************
\subsection{Spaces of histories and non-local observables}
We now have introduced all the mathematical technology
necessary to define the spaces of histories of a field
theory and the notion of a non-local observable properly.

\begin{definition}
Let $\pi:C\to M$ be a morphism of supervarieties modeled on the category
$\Legos=\Alg_{s}^{op}$ or $\Alg_{s,scg}^{op}$, with the topology $\tau_{Zar}$.
A space of histories for $\pi$ is a subspace
$H$ of the space $\uGamma(M,C)$ whose points with values in a superalgebra
$A$ are given by
$$\uGamma(M,C)(A):=\Gamma_{\uSpec(A)}(M\times \uSpec(A),C\times \uSpec(A)).$$
A non-local observable is a morphism
$$F:H\to B$$
with values in another space $B$ of the same type.
\end{definition}

Abstract observables are just a weak mathematical version of what physicists
call observable in DeWitt's covariant field theory \cite{Dewitt}.
The space of values of a non-local observable is usually simply $B=\R=\A^{1|0}$.

\begin{comment}
Let $(\Aff^s,\tau_{Zar})$ be the category dual to the category of real superalgebras
equipped with the Zariski topology.
Let $M$ be a supervariety (that is usually supposed to be a usual
variety, called spacetime) and $C\to M$ be a smooth $M$-space.
\end{comment}

Recall that $M$ and $C$ are defined as functors
$\underline{C},\underline{M}:\Aff_{s,\R}^{op}\to \Sets$
that are sheaves for the Zariski topology.
The basic cases used to construct field theories in experimental physics are following two.
\begin{itemize}
\item Bosonic field theory: $M$ is a four-dimensional (non super) variety and $C$ is a usual fiber bundle over $M$ (for example a connection bundle or a vector bundle). For example,
the interaction particles are usually given by connections and the Higgs boson is a
scalar field (usual function with $C=M\times \R$). 
\item Fermionic field theory: $M$ is a $4$-dimensional variety and $C$ is an odd (spinorial)
vector bundle over $M$.
The algebra of ``functions'' on $C$ is $\Cc^\infty(C):=\Gamma(M,\wedge C^*)$, i.e.
``functions'' on $C$ that are smooth on $M$ and antisymmetric on the fibers of $C\to M$.
\end{itemize}
In theoretical physics, there are many more possibilities. For example:
\begin{itemize}
\item Fermionic particles in spacetime: $M=\R^{0|1}$ and $C=M\times X$ with $X$ a four-dimensional
Lorentz manifold. Dirac's first quantization of the electron rests on this sound basis. 
\item Supersymmetric sigma models: $M$ is a supervariety obtained by adding odd coordinates
to a given classical variety $|M|$ that is usually a Riemann surface and $C$ is of the form
$M\times X$ with $X$, say a Calabi-Yau manifold. This is the starting point of superstring theory
and of many interesting mathematical applications (Gromov-Witten theory, mirror symmetry,
Topological quantum field theory).
\end{itemize} 

In all cases, if $A$ is a superalgebra, one defines the points of $C$ and $M$
with values in $A$ as
$\underline{C}(A)=\Hom_{s\Alg}(\Cc^\infty(C),A)$ and
$\underline{M}(A)=\Hom_{s\Alg}(\Cc^\infty(M),A)$.
We remark for example that for $C$ an odd vector bundle, a function
$$f:C\to \R=\A^{1|0}$$
is simply an element of $\Gamma(M,\wedge^{2*}C)$.

\begin{comment}
\begin{definition}
The space of histories of a general map $C\to M$ is simply the space
$\uGamma(M,C)$ of local sections of the projection:
it is the sheaf on $(\Aff^s,\tau_{Zar})$ associated to the presheaf on $\Aff^s$ given by
$$\Spec(A)\mapsto \Gamma_{\Spec(A)}(M\times \Spec(A),C\times \Spec(A)).$$
An abstract observable is a morphism
$$F:\uGamma(M,C)\to Y$$
from the space of fields to any other (super)space.
\end{definition}
\end{comment}

In the case of a bosonic field theory, $C$ and $M$ being usual spaces, one can
restrict $\uGamma(M,C)$ to $\Alg_{\R}$, $\Alg_{sa}$, or even to $\Open_{\Cc^\infty}$,
depending on the needs. Concretely, the functor
$\uGamma(M,C):\Open_{\Cc^\infty}\to \Sets$
sends an open subset $U$ in some $\R^n$ to the families
$s:U\times M\to C$ of sections of $C\to M$ parameterized by $U$.
This is the original diffeological approach of Souriau to spaces of maps,
and this approach is very close to the way physicists compute.

An $\R$-valued observable is then a natural transformation
$$F:\uGamma(M,C)(.)\to \Hom_{\Cc^\infty}(.,\R)$$
that sends parameterized families $s:U\times M\to C$ of sections to
parameterized real numbers, i.e., to functions $r:U\to \R$. In particular,
for $U=\{pt\}$, we associate to each section $s:M\to C$ (field) a
real number $F(s)\in \R$. 

One can also consider observable of evaluation at a point $x\in M$,
which sends $s:M\times U\to C$ to $s(x,.):U\to C$, denoted
$$\ev_{x}:\uGamma(M,C)\to C.$$
It actually takes its values in the fiber $C_{x}$ of $C$ at $x$.
If we suppose that $C$ is a supervector bundle, we can make sense of
the standard observables whose mean values give correlation
functions, by taking the formal product $\ev_{x}\ev_{y}$
for $x$ and $y$ two evaluation observables. It takes values in
the space algebra
$$\Lambda_{\infty,C}:=\Sym_{s}^\bullet(\oplus_{x\in M}C_{x}),$$
where the functor $\Sym_{s}$ denotes the symmetric algebra
taken in the super sense.
 
The relation of our description of observables with the ``polynomial observables''
of Costello \cite{Costello} is the following.
The functorial version of Costello's observables is given by the space algebra
$$\Oc_{C}:=\oplus_{n\geq 0} \Hom_{\A^1-\Mod}(\Gamma(M^n,\boxtimes^n C),\A^1).$$
If $A\in \Oc_{C}$ is an element, one defines the corresponding $\A^1$-valued observable
$$
\begin{array}{ccccc}
A: 	& \uGamma(M,C)& \to 	 & \A^1\\
	& \phi			&\mapsto& \sum_{n\geq 0}A_{n}(\phi\boxtimes \dots\boxtimes \phi).
\end{array}
$$
It can be useful to work directly with $A\in \Oc_{C}$ as a multilinear map
$$A:\oplus_{n\geq 0}\Gamma(M,C)^{\otimes n}\to \A^1,$$
where the tensor product and direct sums are made in the category of space $\A^1$-modules.

Another type of observable is given by local functionals, which we will study in
detail in section \ref{local-functional-calculus}:
if $J^\infty(C)$ denotes the space of infinite jets of sections
of $C\to M$ with coordinates $(x,u,u_{\alpha})$ representing
formal derivatives of sections, a function $L(x,u,u_{\alpha})\in \Cc^\infty(J^\infty(C),\R)$
with $|\alpha|\leq k$ defines a horizontal differential form $\omega=Ld^nx$
on $J^\infty C$. If $s:M\to C$ is a section of $C$, its infinite jet is a section
$j_{\infty}s:M\to J^\infty C$ and the pull-back of $\omega$ along this section
is an $n$ form on $M$ that can be integrated. If one fixes
a compact domain $K\subset M$, one defines
\footnote{Remark that this only defines a partial function, whose domain can be defined
by Lebesgue's domination condition on its integrand. This condition is functorial in $U$ and defines
a subspace of the space of sections. This remark applies to all integrals written in this paper.}
an observable $S_{L,K}$ on sections
with support in $K$, with values in $\R$, called the lagrangian action by
$$
\begin{array}{cccc}
S_{L}: & \uGamma_{K}(M,C)(U) & \to & \Cc^\infty(U,\R)\\
&[s:U\times M\to C] & \mapsto  & [u\mapsto \int_{M}j_{\infty}(s(u,.))^*\omega].
\end{array}
$$

More generally, if $Y\subset J^\infty C$ corresponds to a differential equation
(see the forthcoming section on $\Dc$-schemes), there is a natural period pairing
$$H_{k,c}(M)\times \bar{H}^k_{dR}(Y)\to \underline{\Hom}(\underline{Y},\underline{\R})$$ 
(where $\underline{Y}\subset \uGamma(M,C)$ is the subspace defined by $Y$)
given by the same formula
$$
(\alpha,\omega)\mapsto
\left[\phi(u,x)\mapsto \int_{\alpha}(j_{\infty}\phi(u,.))^*\omega\right]$$
whose image is called the space of secondary functions on $\underline{Y}$
with values in $\R$.

It is in the case of fermionic field theory that one sees the real input
of the functorial viewpoint of spaces of histories. Indeed, the usual
notion of a point of a supervariety is not well behaved and one really
has to use points valued in superalgebras. If $M$ is usual spacetime
and $C\to M$ is an odd fiber bundle, the usual duality between spaces and algebras
implies that points of the space of fermionic
histories $\uGamma(M,C)$ with values in $\{pt\}=\Spec(\R)$
are given by retractions $s^*:\Cc^\infty(C)\to \Cc^\infty(M)$ of the
map $\pi^*:\Cc^\infty(M)\to \Cc^\infty(C)$ that induces the projection
$$
\pi:\underline{C}=\Hom_{s\Alg}(\Cc^\infty(C),.)\to \Hom_{s\Alg}(\Cc^\infty(M),.)=\underline{M}.
$$
Since $\Cc^\infty(C)$ is partially antisymmetric and $\Cc^\infty(M)$ is commutative,
such retractions will have to be trivial on odd coordinates. If one replaces
$\{pt\}$ by the superspectrum of an odd algebra $\Spec(A)$, the parameterized retractions
$s^*:\Cc^\infty(C)\to \Cc^\infty(M)\otimes A$ can be much more general.
DeWitt in his enormous book \cite{Dewitt} has chosen to use the completion
of the free odd algebra on a countable number of generators
$$A=\Lambda_{\infty}:=\widehat{\wedge^*\R^{(\N)}}$$
to study fermionic fields,
which is sufficient for most of the computations needed with fermionic
functional integrals, but there is no physical reason to choose this algebra or
another one.
In any case, an $\A^{1,1}_{\R}$-valued abstract observable
$$F:\uGamma(M,C)(.)\to \A^{1,1}_{\R}(.)$$
will associate to each retraction $s^*$ of $\pi^*$ a real supernumber parameterized
by a given superalgebra $A$, which is the same as an element of $A$ because
$\A^{1,1}_{\R}(A)=A$. This means that a fermionic observable is essentially determined
(and this is how DeWitt formalizes it) by its $\Lambda_{\infty}$ values
$$
F_{\Lambda_{\infty}}:
\uGamma(M,C)(\Lambda_{\infty})\to \A^{1,1}(\Lambda_{\infty})=\Lambda_{\infty}.
$$

To convince the reader, let us give a further example of a trajectory with fermionic
parameter, which is at the basis of Dirac's quantization of the electron.

\begin{proposition}
Let $X$ be a smooth variety, seen as a superspace. Denote
$\pi:C=X\times \R^{0|1}\to \R^{0|1}=M$ the natural projection map.
It is the configuration space for the so-called fermionic particle.
There is a natural isomorphism of functors on $\Alg_{s,scg}$
$$\uGamma(M,C)\cong \uHom(\R^{0|1},M)\cong \uSpec(\Omega^*(X)).$$
In particular, the real-valued functions on $\uGamma(M,C)$, given by morphisms
$$\uGamma(M,C)\to \R=\A^{1|0},$$
identify with even differential forms in $\Omega^{2*}(X)$.
\end{proposition}
\begin{proof}
Let $A$ be a superalgebra. The first isomorphism follows from
the trivial bundle structure of $\pi:C\to M$. The space
$\uHom(\R^{0|1},X)$ is defined by
$$\uHom(\R^{0|1},X)(A):=\Hom_{\Alg_{s,g}}(\Cc^\infty(X),\R[\theta]\otimes A)$$
where $\theta$ is an anticommuting variable. We remark that if
$A=\R$, we get the usual set of morphisms $\Hom(\R^{0|1},X):=\Hom_{\Alg_{s,g}}(\Cc^\infty(X),\R[\theta])$, which identifies, since $\theta$ is odd, with $\Hom_{\Alg_{s,g}}(\Cc^\infty(X),\R)=X$.
The main advantage of adding an odd parameter algebra $A$ is that the even part
of $\R[\theta]\otimes A$ is $A^0\oplus\R.\theta\otimes A^1$.
%For example, if $A=\R[\theta_{2}]$,
%one gets as even part $\R\oplus \R.\theta\theta_{2}$. A superalgebra morphism
%$f:\Cc^\infty(M)\to \R[\theta]\otimes A$
Let $f^*:\Cc^\infty(X)\to \R[\theta]\otimes A$ be a morphism. Then
$f^*$ can be written as $f_{0}+\theta f_{1}$, where $f_{0}:\Cc^\infty(X)\to A^0$ is
a usual morphism and $f_{1}:\Cc^\infty(X)\to A^1$ is a derivation (because $\theta^2=0$)
compatible with the $\Cc^\infty(X)$-module structure on $A^1$ induced by $f_{0}$
and the multiplication in $A$. Since $A^1$ is a geometric module,
such a derivation can be identified with a
$\Cc^\infty(X)$-module morphism $\Omega^1(X)\to A^1$.
This identifies with a superalgebra morphism $\Omega^*(X)\to A$.
\end{proof}

Let us explain why this fermionic particle is so important in physics.
In Dirac's first quantization of the electron, one considers the Clifford algebra
$\Cliff(TX,g)$ for a given lorentzian metric $g$ on $X$ as a quantization of the fermionic particle
$$x:\R^{0|1}\to X,$$
because the Clifford algebra has a filtration $F$ such that
$$\gr^\bullet_{F}\Cliff(TX,g)\cong \Omega^*_{X}$$
and the commutator in the Clifford algebra corresponds by this isomorphism
to the (odd) Poisson bracket on
$$\Omega^*_{X}\overset{g^{-1}}{\overset{\sim}{\to}} \wedge^*\Theta_{X}$$
Its state space is the space $\Gamma(M,S)$ of sections of a spinor bundle $S$ for $g$,
supposed to exist.

\section{Local observables and differential schemes}
\label{local-functional-calculus}
We will now give an account of the differential calculus on local
action functionals. A more general differential calculus on
local functionals, called secondary differential calculus,
was developed in the smooth setting by Vinogradov
\cite{Vinogradov}.
It is essentially a homotopical version of the following.

%************************************************************************
\subsection{Partial differential equations and $\Dc$-algebras}
Many spaces of functions used in physics are described by
partial differential equations. To study spaces of solutions of
partial differential equations from a functorial viewpoint,
one needs to know what kind of algebraic structure is necessary
to write down a given partial differential equation.

Let $\pi:C=\R\times \R\to \R=M$ be the trivial bundle.
A polynomial partial differential equation on the sections
of this bundle is a polynomial expression
$$F(t,x,\partial_{t}x,\dots,\partial_{t}^nx)=0$$
that involves the parameter $t\in M$, a section $x:M\to C$ and
its derivatives. To write down the same expression in a more
general algebraic structure, one needs:
\begin{itemize}
\item an $\Oc_{M}=\R[t]$-algebra $A$,
\item with a compatible action of $\partial_{t}$, and
\item a morphism $\Oc_{C}=\R[t,x]\to A$.
\end{itemize}
One can also see this datum as a $\Dc_{M}$-module $A$ (where $\Dc_{M}$
is the algebra of differential operators on $M$, generated by the action
of $\partial_{t}$ and $\Oc_{M}$ on $\Oc_{M}$), equipped with a multiplication
$$\mu:A\otimes_{\Oc_{M}}A\to A$$
that is $\Dc_{M}$-linear, for the $\Dc_{M}$-module structure on the tensor
product given by Leibniz's rule. One moreover needs a morphism
$$\Oc_{C}\to A$$
to make sense of $x\in A$. If such a datum is given, one writes
the solution space to $F(t,\partial_{t}^i x)=0$ with values in $A$ as
$$\Sol_{F=0}(A):=\{x\in A,F(t,\partial_{t}^i x)=0\}.$$
One sees here a strong similarity with the space of solutions of
a polynomial equation, described in \ref{polynomial-solution}.
The point is that the equation $F$ itself lives in the universal
$\Oc_{C}\otimes_{\Oc_{M}}\Dc_{M}$-algebra, that is the jet
$\Dc_{M}$-algebra $\Jet(\Oc_{C})=\R[t,x_{i}]$ with action of
$\partial_{t}$ given by $\partial_{t}x_{i}=x_{i+1}$.

So one gets a perfect analogy between polynomials and polynomial partial
differential equations given by
\vskip0.5cm
\begin{tabular}{|c|c|c|}\hline
Equation	& Polynomial & Partial differential\\\hline
Formula 	& $P(x)=0$	 & $F(t,\partial^\alpha x)=0$\\
Naive variable & $x\in \R$ & $x\in \Hom(\R,\R)$\\
Algebraic structure & commutative unitary ring $A$   & $\Dc_{M}$-algebra $A$\\
Free structure & $P\in \R[x]$	& $F\in \Jet(\Oc_{C})$\\
Solution space & $\{x\in A,F(x)=0\}$ & $\{x\in A,F(t,\partial^\alpha x)=0\}$\\\hline
\end{tabular}

To work with non-polynomial smooth partial differential equations, one has
to work in the category $\Alg_{scg}$ of smoothly closed geometric algebras.
In this setting, the jet algebra is the smooth geometric closure of the polynomial
jet algebra and the equation $F$ lives in
$$\Cc^\infty(J^\infty C):=\overline{\Jet(\Oc_{C})}^{scg}.$$

%************************************************************************
\subsection{$\Dc$-modules}
We recall some properties of the categories of $\Dc$-modules that can mostly be found
in \cite{Kashiwara-Schapira-sheaves-on-manifolds}, \cite{Kashiwara-D-modules} and \cite{Schneiders}
for most of them, except the compound tensor structure, which was defined in
\cite{Beilinson-Drinfeld-Chiral}.

Let $M$ be a smooth variety of dimension $n$ and $\Dc$ be the algebra of
differential operators on $M$.
We recall that, locally on $M$, one can write an operator $P\in \Dc$
as a finite sum
$$P=\sum_{\alpha}a_{\alpha}\partial^\alpha$$
with $a_{\alpha}\in \Oc_{M}$,
$$\partial=(\partial_{1},\dots,\partial_{n}):\Oc_{M}\to \Oc_{M}^n$$
the universal derivation and $\alpha$ some multi-indices.

To write down the equation $Pf=0$ with $f$ in an $\Oc_{M}$-module $\Sc$,
one needs to define the universal derivation $\partial:\Sc\to \Sc^n$. This is equivalent
to giving $\Sc$ the structure of a $\Dc$-module. The solution space
of the equation with values in $\Sc$ is then given by
$$\Sol_{P}(\Sc):=\{f\in \Sc,\;Pf=0\}.$$
We remark that
$$\Sol_{P}:\Mod(\Dc)\to \Vect_{\R_{M}}$$
is a functor that one can think of as representing the space of
solutions of $P$. Denote $\Mc_{P}$ the cokernel of the $\Dc$-linear map
$$\Dc\overset{.P}{\longrightarrow}\Dc$$
given by right multiplication by $P$.
Applying the functor $\Homc_{\Mc(\Dc)}(.,\Sc)$ to the exact sequence
$$\Dc\overset{.P}{\longrightarrow}\Dc\longrightarrow \Mc_{P}\to 0,$$
we get the exact sequence
$$0\to \Homc_{\Mod(\Dc)}(\Mc_{P},\Sc)\to \Sc\overset{P.}{\longrightarrow} \Sc,$$
which gives a natural isomorphism
$$\Sol_{P}(\Sc)=\Homc_{\Mod(\Dc)}(\Mc_{P},\Sc).$$
This means that the $\Dc$-module $\Mc_{P}$ represents
the solution space of $P$, so that $\Dc$-modules are a convenient
setting for the functor of point approach to linear partial differential equations.

We remark that it is even better to consider the derived solution space
$$\R\Sol_{P}(\Sc):=\R\Homc_{\Mod(\Dc)}(\Mc_{P},\Sc)$$
because it also encodes information on the inhomogeneous
equation
$$Pf=g.$$

Recall that the subalgebra $\Dc$ of $\End_{\R}(\Oc)$
is generated by left multiplication by functions in $\Oc_{M}$
and by the derivation induced by vector fields in $\Theta_{M}$.
There is a natural right action of $\Dc$ on the $\Oc$-module $\Omega^n_{M}$
by
$$\omega.\partial=-L_{\partial}\omega$$
with $L_{\partial}$ the Lie derivative.

There is a tensor product in the category $\Mod(\Dc)$ given by
$$\Mc\otimes \Nc:=\Mc\otimes_{\Oc}\Nc$$
where the $\Dc$-module structure on the tensor product is given
on vector fields $\partial\in \Theta_{M}$ by Leibniz's rule
$$\partial(m\otimes n)=(\partial m)\otimes n+m\otimes (\partial n).$$
There is also an internal homomorphism $\Homc(\Mc,\Nc)$ given by
the $\Oc$-module $\Homc_{\Oc}(\Mc,\Nc)$ equipped with the action
of derivations $\partial\in \Theta_{M}$ by
$$\partial(f)(m)=\partial(f(m))-f(\partial m).$$
The functor
$$\Mc\mapsto \Mc^r:=\Omega^n_{M}\otimes_{\Oc}\Mc$$
induces an equivalence of categories between the categories
$\Mod(\Dc)$ and $\Mod(\Dc^{op})$ of left and right $\Dc$-modules
whose quasi-inverse is
$$\Nc\mapsto \Nc^\ell:=\Homc_{\Oc_{M}}(\Omega^n_{M},\Nc).$$

\begin{definition}
Let $\Sc$ be a right $\Dc$-module.
The de Rham functor with values in $\Sc$ is the functor
$$\DR_{\Sc}:\Mod(\Dc)\to \Vect_{\R_{M}}$$
that sends a left $\Dc$-module to
$$\DR_{\Sc}(\Mc):=\Sc\Lotimes_{\Dc}\Mc.$$
The de Rham functor with values in $\Sc=\Omega^n_{M}$
is denoted $\DR$ and simply called the de Rham functor.
One also denotes $\DR^r_{\Sc}(\Mc)=\Mc\Lotimes_{\Dc}\Sc$
if $\Sc$ is a fixed left $\Dc$-module and $\Mc$ is a varying
right $\Dc$-module, and $\DR^r:=\DR^r_{\Oc}$.
\end{definition}

\begin{proposition}
\label{universal-de-Rham}
The natural map
$$
\begin{array}{ccc}
\Omega^n_{M}\otimes_{\Oc}\Dc&\to& \Omega^n_{M}\\
\omega\otimes Q&\mapsto&\omega Q
\end{array}
$$
extends to a $\Dc^{op}$-linear quasi-isomorphism
$$\Omega^*_{M}\otimes_{\Oc}\Dc[n]\overset{\sim}{\to}\Omega^n_{M}.$$
\end{proposition}

We will see that in the super setting, this proposition can be taken
as a definition of the right $\Dc$-modules of volume forms, the so called
Berezinian.

\begin{proposition}
Let $\Sc$ be a right coherent $\Dc$-module and $\Mc$ be a coherent left $\Dc$-module.
There is a natural quasi-isomorphism
$$\R\Sol_{\Db(\Mc)}(\Sc):=\R\Hom(\Db(\Mc),\Sc)\cong\DR_{\Sc}(\Mc),$$
where $\Db(\Mc):=\R\Hom(\Mc,\Dc)$ is the $\Dc$-module dual of $\Mc$.
\end{proposition}

%************************************************************************
\subsection{$\Dc$-modules on supervarieties and the Berezinian}
\label{Berezinian}
We refer to Penkov's article \cite{Penkov} for a complete study
of the Berezinian in the $\Dc$-module setting.

Let $M$ be a supervariety of dimension $n|m$.
As explained in subsection \ref{diff-calc-sym-monoidal}
one defines $\Omega^1_{M}$ as the
representing object for the internal derivation functor $\underline{Der}(\Oc_{M},.)$
on geometric $\Oc_{M}$-modules. One also defines $\Omega^*_{M}$
as the superexterior power
$$\Omega^*_{M}:=\wedge^*\Omega^1_{M}.$$

The super version of Proposition \ref{universal-de-Rham}
can be taken as a definition of the Berezinian,
as a complex of $\Dc$-modules, up to quasi-isomorphism.
\begin{definition}
The Berezinian of $M$ is defined
in the derived category of $\Dc_{M}$-modules by the formula
$$\Ber_{M}:=\Omega^*_{M}\otimes_{\Oc}\Dc[n].$$
The complex of integral forms $I_{*,M}$ is defined by
$$I_{*,M}:=\R\Hom_{\Dc}(\Ber_{M},\Ber_{M}).$$
\end{definition}

The following proposition (see \cite{Penkov}, 1.6.3) gives a description of the Berezinian
as a $\Dc$-module.
\begin{proposition}
The Berezinian complex is concentrated in degree $0$ and equal there
to
$$\Ber_{M}:=\Extc^n_{\Dc}(\Oc,\Dc).$$
The functor
$$
\begin{array}{ccc}
\Mod(\Dc) 	& \to 		& \Mod(\Dc^{op})\\
\Mc			& \mapsto 	& \Mc^r:=\Ber_{M}\otimes \Mc 
\end{array}
$$
is an equivalence of categories with quasi-inverse $\Mc\mapsto \Mc^\ell:=\Ber_{M}^{-1}\otimes \Mc$.
This equivalence and the tensor product of left $\Dc$-modules over $\Oc$
induce a monoidal structure on $\Mod(\Dc^{op})$, denoted $\otimes^!$.
\end{proposition}

In the supersetting, the equivalence of left and right $\Dc$-modules,
given by the functor
$$\Mc\mapsto \Ber_{M}\otimes_{\Oc}\Mc$$
of twist by the Berezinian right $\Dc$-module,
can be computed by using the definition
$$\Ber_{M}:=\Omega^*_{M}\otimes_{\Oc}\Dc[n]$$
and passing to degree $0$ cohomology.

A more explicit description of the complex of integral forms (up to
quasi-isomorphism) is given by
$$
I_{*,M}:=\R\Homc_{\Dc}(\Ber_{M},\Ber_{M})
\cong
\Homc_{\Dc}(\Omega^*_{M}\otimes_{\Oc}\Dc[n],\Ber_{M})
$$
so that we get
$$
I_{*,M} \cong
\Homc_{\Oc}(\Omega^*_{M}[n],\Ber_{M})
\cong \Homc_{\Oc}(\Omega^*_{M}[n],\Oc)\otimes_{\Oc}\Ber_{M}
$$
and in particular $I_{n,M}\cong \Ber_{M}$.

We remark that proposition \ref{universal-de-Rham} shows that
if $M$ is a usual variety, then $\Ber_{M}$ is quasi-isomorphic
with $\Omega^n_{M}$, and this implies that
$$
I_{*,M}\cong \Homc_{\Oc}(\Omega^*_{M}[n],\Oc)\otimes_{\Oc}\Ber_{M}
\cong \wedge^*\Theta_{M}\otimes _{\Oc} \Omega^n_{M}\overset{i}{\longrightarrow}
\Omega^*_{M},
$$
where $i$ is the insertion morphism. This implies the isomorphism
$$
I_{n-p,M}\cong \Omega^p_{M},
$$
so that in the purely even case, integral forms essentially identify with
usual differential forms.

We recall from Bernstein and Leites' work \cite{Bernstein-Leites} that,
with a convenient notion of compactly supported homology
$H_{*,c}(M)$ on a given supermanifold $M$, there is an integration pairing
$$H_{*,c}(M)\times h^*(I_{*,M})\to \R$$
that reduces to the usual integration pairing
$$H_{*,c}(M)\times H^*_{dR}(M)\to \R$$
in the classical case.
The integration of an integral form in $I_{n-m-p}$ is done on subsupermanifolds
of dimension $p|m$ of a given supermanifold of dimension $n|m$.

%************************************************************************
\subsection{$\Dc$-algebras and partial differential equations}
In this subsection, we will work with varieties modeled on the category
$\Alg_{\R}$, $\Alg_{scg}$, $\Alg_{s}$ or $\Alg_{s,scg}$.

The tensor structure on the category of left $\Dc$-modules allows us
to define $\Dc$-algebras and $\Dc$-schemes using the philosophy of
relative geometry in monoidal categories, like in section \ref{relative-geometry}.
\begin{comment}
 first introduced by Deligne in \cite{De}
(see also Hakim's thesis \cite{Hakim}) and further studied by Toen and
Vaqui\'e \cite{Toen-Vaquie}.
\end{comment}

We recall that the $\Dc$-module structure
on $\Mc\otimes_{\Oc}\Nc$ is given by Leibniz's rule
$$D(a\otimes b)=Da\otimes b+a\otimes Db$$
on the level of derivations. This means that the following notion of $\Dc$-algebra
is just a $\Dc$-module equipped with a multiplication that fulfils Leibniz rule for
derivations.

\begin{definition}
A (commutative) $\Dc$-algebra is a monoid in the monoidal category
of (left) $\Dc$-module, i.e.,
it is a $\Dc$-module $\Ac$ equipped with a multiplication morphism of $\Dc$-modules
$$\mu:\Ac\otimes_{\Oc_{X}}\Ac\to \Ac$$
and a unit map $1:\Oc_{X}\to \Ac$ that fulfil associativity, unity and commutativity
axioms. If $\Ac$ is a $\Dc$-algebra, a module over $\Ac$ is given by a $\Dc$-module
$\Mc$ and a morphism $\mu:\Ac\otimes \Mc\to \Mc$ of external multiplication that is
compatible with unit and multiplication on $\Ac$ in the usual sense.
\end{definition}

\begin{comment}
\begin{proposition}
If $\Ac$ is a $\Dc$-algebra, the category $\Mod(\Ac[\Dc])$ of modules
over the algebra $\Ac[\Dc]$ (also called horizontal differential modules
on $\Ac$), equipped with the tensor product $\Mc\otimes\Nc:=\Mc\otimes_{\Ac}\Nc$
is a symmetric monoidal category with unit object $\Ac$.
\end{proposition}
\end{comment}

\begin{definition}
A $\Dc$-space is a sheaf on the site $(\Alg_{\Dc}^{op},\tau_{Zar})$ of $\Dc$-algebras
with their Zariski topology.
\end{definition}

We now introduce the differential algebraic analog of polynomial algebra,
called jet algebra, by recalling the following result \cite{Beilinson-Drinfeld-Chiral}, 2.3.2.
\begin{proposition}
Let $\pi:C\to M$ be a smooth map between varieties. There exists a free
$\Dc_{M}$-algebra generated by $\Oc_{C}$, denoted $\Jet(\Oc_{C})$.
More precisely, one has, for every $\Dc$-algebra $\Ac$, a natural isomorphism
$$\Hom_{\Dc_{M}-\Alg}(\Jet(\Oc_{C}),\Ac)\cong \Hom_{\Oc_{M}-\Alg}(\Oc_{C},\Ac).$$
Its spectrum is denoted $\Jet(\pi)$, or simply $\Jet(C)$.
\end{proposition}
\begin{proof}
The algebra $\Jet(\Oc_{C})$ is given by the quotient of the symmetric algebra
$$\Sym^\bullet(\Dc\otimes_{\Oc_{X}}\Oc_{C})$$
by the ideal generated by the elements
$\partial(1\otimes r_{1}.1\otimes r_{2}-1\otimes r_{1}r_{2})\in \Sym^2(\Dc\otimes_{\Oc_{X}}R)+
\Dc\otimes_{\Oc_{X}}R$ and $\partial(1\otimes 1_{R}-1)\in \Dc\otimes R+\Oc_{X}$,
$r_{i}\in R$, $\partial\in \Dc$, $1_{R}$ the unit of $R$.
\end{proof}

We remark that the $\Jet$ algebra is not in general finitely generated as an algebra over $\Oc_{M}$,
but by definition, it is finitely generated as a $\Dc_{M}$-algebra if $\Oc_{C}$ is finitely generated
over $\Oc_{M}$. If $s:M\to C$ is a section, we denote $j_{\infty}s:M\to \Jet(C)$ the corresponding
map with values in the jet space.

We have defined here only the algebraic jet space, but one can define the usual jet space
by working with (super)algebras that are smoothly closed geometric.
Indeed, the smooth closure of the algebraic jet algebra gives the algebra
of functions on usual infinite jet space, by construction. Our methods thus apply also
to the smooth case, if one works with the convenient category of algebras and modules
over them as in \cite{Nestruev}.

We now give a definition of a partial differential equation and of its spaces of solutions,
which works equally well in the smooth, algebraic or supergeometric setting.
\begin{definition}
A partial differential equation on the sections of $\pi:C\to M$ is a $\Dc$-ideal $\Ic$ in $\Jet(\Oc_{C})$.
Its differential solution space is the $\Dc$-space whose points with values in
$\Jet(\Oc_{C})$-$\Dc$-algebras $\Ac$ are
$$\Sol_{\Dc,\Ic=0}(\Ac)=\{x\in \Ac,\;f(x)=0,\;\forall f\in \Ic\}.$$
Its solution space is the subspace of $\uGamma(M,C)$ whose points with values
in a test algebra $A$ (in $\Alg_{scg}$, $\Alg_{\R}$ or $\Alg_{s,scg}$) are given by
$$\Sol_{\Ic=0}(A)=\{s(t,u)\in \uGamma(M,C)(A),\;f\circ (j_{\infty,t}s)(t,u)=0,\;\forall f\in \Ic\},$$ 
where we identify $\uGamma(M,C)(A)$ with a subset $\Hom(M\times \uSpec(A),C)$.
\end{definition}

We remark that one has an isomorphism of $\Dc$-spaces
$$\Sol_{\Dc,\Ic=0}\cong \uSpec_{\Dc}(\Jet(\Oc_{C})/\Ic)$$
which means that the differential solution space is in some sense
(which will be clarified in the next subsection) finite dimensional.
This is very different of the diffeological solution space that is far away from
being a finite-dimensional manifold in general. This finite dimensionality
can be seen as the mathematical reason why physicist like to work
with local functionals.

%************************************************************************
\subsection{Local functionals and local differential calculus}
Let $M$ be a supermanifold of dimension $n|m$ and $\Mc$ be a $\Dc_{M}$-module.
We suppose that the underlying manifold $|M|$ is oriented.

\begin{definition}
The central de Rham cohomology of $\Mc$ is defined by
$$h(\Mc):=\Ber_{M}\otimes_{\Dc}\Mc.$$
The variational de Rham complex of $\Mc$ is defined by
$$\DR_{var}(\Mc):=(I_{n-m-*,M}\otimes_{\Oc}\Dc)\otimes_{\Dc}\Mc.$$
\end{definition}

We remark that in the classical case of dimension $n|0$, the variational
de Rham complex identifies with the usual de Rham complex
$$
\DR(\Mc):=\Omega^n_{M}\Lotimes_{\Dc}\Mc=(\Omega^*_{M}\otimes_{\Oc}\Dc[n])\otimes_{\Dc}\Mc
$$
and the central de Rham cohomology is isomorphic to
$$h(\Mc)=\Omega^n_{M}\otimes_{\Dc}\Mc.$$

Let $\pi:C\to M$ be a bundle and $H\subset \uGamma(M,C)$ be a subspace
that is a solution space of a given partial differential $\Ic_{H}$ on $\uGamma(M,C)$.
Let $\Ac=\Jet(\Oc_{C})/\Ic_{H}$ be the corresponding $\Dc$-algebra. We suppose
that it is $\Dc$-smooth.

Integral forms can be integrated on subsupermanifolds. This allows us
to define \footnote{The functionals in play have a domain of definition
given by Lebesgue's domination condition on the integrand.} an integration pairing.

\begin{proposition}
There is a well-defined integration pairing
$$
\begin{array}{ccc}
H_{*,c}(M)  \times  h^*(\DR_{var}(\Ac))&\to& \uHom(H,\R)\\
(\Sigma,\omega) & \mapsto & F_{\Sigma,\omega}:s(t,u)\mapsto \int_{\Sigma}(j_{\infty,t}s(t,u))^*\omega.
\end{array}
$$
\end{proposition}
\begin{proof}
This follows from the fact that the integral of a total derivative is zero.
\end{proof}

\begin{definition}
A functional $F_{\Sigma,\omega}:H\to \R$ in the image of the above pairing is
called a local functional on $H$.
\end{definition}

Let $d:\Ac\to \Omega^1_{\Ac}$ be the universal derivation with values
in $\Ac[\Dc]$-modules ($\Ac$-modules in the tensor category of $\Dc$-modules).
Let $\Omega^*_{\Ac}=\wedge_{\Ac}^*\Omega^1_{\Ac}$ be the corresponding
algebra of differential forms. One can generalize the above notion of
local functional to differential forms.

\begin{proposition}
There is a well-defined integration pairing
$$
\begin{array}{ccc}
H_{*,c}(M)  \times  h^*(\DR_{var}(\Omega^k_{\Ac}))&\to& \Omega^k_{H}\\
(\Sigma,\omega) & \mapsto & \nu_{\Sigma,\omega}:[s:M\times U\to C]\mapsto \int_{\Sigma}(j_{\infty}s)^*\omega\in \Omega^k_{U}.
\end{array}
$$
\end{proposition}

\begin{definition}
A differential form $\nu_{\Sigma,\omega}\in \Omega^k_{H}$ in the image of the
above pairing is called a local differential form on $H$.
\end{definition}

\begin{definition}
Let $\Ac$ be a smooth $\Dc$-algebra (see \cite{Beilinson-Drinfeld-Chiral}, chapter 2).
The $\Ac^r[\Dc^{op}]$-module of local vector fields on $\Ac$ is defined as
the $\Ac[\Dc]$-dual of differential forms, i.e., by the formula
$$\Theta_{\Ac}:=\Homc_{\Ac[\Dc]}(\Omega^1_{\Ac},\Ac[\Dc]).$$
\end{definition}

The finite dimensionality of the $\Dc$-space of solutions of a partial differential equation
can be explained by the following proposition from \cite{Beilinson-Drinfeld-Chiral},
chapter 2.

\begin{proposition}
Let $\Bc=\Jet(\Oc_{C})$ and $p:\Jet(\Oc_{C})\to C$ be the projection map.
The natural map
$$\Bc[\Dc]\otimes_{\Bc}p^*\Omega^1_{C/M}\to \Omega^1_{\Bc}$$
is an isomorphism in the jet case.
The rank of $\Omega^1_{\Bc}$ as a $\Bc[\Dc]$-module is equal to the
rank of $\Omega^1_{C/M}$ as an $\Oc_{C}$-module.
\end{proposition}

\begin{corollary}
\label{tangent-D-module}
Let $\Bc=\Jet(\Oc_{C})$, $\Bc^r:=\Ber_{M}\otimes_{\Oc_{M}}\Bc$ and $p:\Jet(C)\to C$
be the projection map. The natural morphism of
$\Bc^r[\Dc^{op}]$-modules
$$\Theta_{\Bc}\to (p^*\Theta_{C/M})\otimes_{\Bc^r}\Bc^r[\Dc^{op}]$$
is an isomorphism. The rank of $\Theta_{\Bc}$ as a $\Bc^r[\Dc^{op}]$-module
is equal to the rank of $\Theta_{C/M}$ as an $\Oc_{C}$-module.
% The image of a vertical derivation $\chi=f\frac{\partial}{\partial u}\in p^*\Theta_{C/M}$ in $\Theta_{\Bc}$ is denoted $\rE_{\chi}=\sum_{\alpha}D_{\alpha}(f)\frac{\partial}{\partial u_{\alpha}}$
%and called the evolutionary vector field associated to the vertical vector field $\chi$.
\end{corollary}
\begin{proof}
The morphism $p:\Jet(C)\to C$ induces an exact sequence of $\Bc$-modules
$$0\to p^*\Omega^1_{C/M}\to \Omega^1_{\Jet(C)/M}\to \Omega^1_{\Jet(C)/C}\to 0$$
and a natural morphism
$$\Dc\otimes_{\Oc}p^*\Omega^1_{C/M}\to \Omega^1_{\Jet(C)/M}$$
of left $\Dc$-modules (the right-hand side is equipped with its canonical
$\Dc$-module structure). This gives a natural map
$$
\Theta_{\Bc}:=\Homc_{\Bc[\Dc]}(\Omega^1_{\Jet(C)/M},\Bc[\Dc])\to
\Homc_{\Bc[\Dc]}(\Dc\otimes_{\Oc}p^*\Omega^1_{C/M},\Bc[\Dc])
$$
and combining it with the natural isomorphisms
$$
\begin{array}{ccc}
\Homc_{\Bc[\Dc]}(\Dc\otimes_{\Oc}p^*\Omega^1_{C/M},\Bc[\Dc])
& \cong &
\Homc_{\Bc}(p^*\Omega^1_{C/M},\Bc[\Dc])\\
&\cong&
\Homc_{\Bc}(p^*\Omega^1_{C/M},\Bc)\otimes \Dc\\
&=:& (p^*\Theta_{C/M})\otimes\Dc
\end{array}
$$
induces a natural map
$$\Theta_{\Bc}\to (p^*\Theta_{C/M})\otimes\Dc^{op}.$$
The fact that it is an isomorphism comes from the fact that
the natural map
$$\Bc[\Dc]\otimes_{\Bc}p^*\Omega^1_{C/M}\to \Omega^1_{\Bc}$$
is an isomorphism in the jet case.
\end{proof}

\begin{comment}
\begin{proposition}
Denote $\Bc=\Jet(\Oc_{C})$.
Suppose that we work with algebraic superspaces, modeled on $\Alg_{s,\R$.
There is a natural map
$$
h(\Ber_{M}^{-1}\otimes \Theta_{\Bc})\to\uGamma(\uGamma(M,C),T\uGamma(M,C))
$$
that sends a vector field on the differential space $\Spec_{\Dc}(\Bc)$ to a vector
field on $\uGamma(M,C)$.
\end{proposition}
\begin{proof}
One has
$$
h(\Ber_{M}^{-1}\otimes \Theta_{\Bc})\cong
\Theta_{\Bc}\otimes _{\Dc}\Oc\cong \Theta_{\Bc}/\Theta_{X}.\Theta_{\Bc}.
$$
Let $v=f\frac{\partial}{\partial u}$ be a local element in $p^*\Theta_{C/M}$.
Let $\phi$ be a field in $\uGamma(M,C)(\R)=\Gamma(M,C)$.
Recall that
$$T\uGamma(M,C)(\R)=\uGamma(M,C)(\R[\epsilon]/(\epsilon^2))$$
is the space of maps
$$V:\Bc:=\Oc_{C}\to \Oc_{M}[\epsilon]/(\epsilon^2)$$
that are retractions of the natural injection $\Oc_{M}[\epsilon]/(\epsilon^2)\to
\Oc_{C}[\epsilon]/(\epsilon^2)$.
We define such a retraction by the formula
$$
V_{v}(h):=
h\circ \phi+\left(\rE_{v}.h)\circ \phi)\right).\epsilon
$$
which makes sense since its coefficients are functions on $M$. One can easily
write down these formulas with a field $\phi_{t}$ parameterized by the spectrum
$\uSpec(A)$ of a superalgebra.
\end{proof}
\end{comment}

%************************************************************************
\subsection{Variational calculus}
We now recall our general notion of variation problem.
We use here superspaces modeled on geometric superalgebras.

\begin{definition}
A lagrangian variational problem is composed of the following data:
\begin{enumerate}
\item a space $M$ called the parameter space for trajectories,
\item a space $C$ called the configuration space for trajectories,
\item a morphism $\pi:C\to M$ (often supposed to be surjective),
\item  a subspace $H\subset \uGamma(M,C)$ of the space of sections of $\pi$
$$\uGamma(M,C):=\{x:M\to C,\;\pi\circ x=\id\},$$
called the space of histories, and
\item a functional $S:H\to \R$ called the action functional.
\end{enumerate}
The space of classical trajectories for the variational problem is
the subspace $T$ of $H$ defined by
$$T=\{x\in H|\;d_{x}S=0\}.$$
If $B$ is another space, a classical $B$-valued observable is a functional
$F:T\to B$ and a quantum $B$-valued observable is a functional
$F:H\to B$.
\end{definition}

Virtually every example of variational calculus that can be found in the 
classical physical literature is of the following type.
\begin{definition}
A variational problem is called local if the following hold.
\begin{enumerate}
\item The space of histories $H\subset \uGamma(M,C)$ is defined by a differential equation
$\Ic_{H}\subset \Jet(\Oc_{C})$, such that $\Jet(\Oc_{C})/\Ic_{H}$ is $\Dc$-smooth,
\item The action functional is the local functional associated to a cohomology class
$S\in h(\Ac)$ for $\Ac=\Jet(\Oc_{C})/\Ic_{H}$.
\end{enumerate}
\end{definition}

Suppose that we are given a local variational problem.
We suppose that $\Ac$ is $\Dc$-smooth (see \cite{Beilinson-Drinfeld-Chiral}).
Using the biduality isomorphism
$$\Omega^1_{\Ac}\cong \Homc_{\Ac^r[\Dc^{op}]}(\Theta_{\Ac},\Ac^r[\Dc^{op}]),$$
one gets a well-defined isomorphism
$$h(\Omega^1_{\Ac})\cong \Homc_{\Ac[\Dc]}(\Ber_{M}^{-1}\otimes\Theta_{\Ac},\Ac).$$
To the given action functional $S\in h(\Ac)$ corresponds its differential
$h(d)(S)\in h(\Omega^1_{\Ac})$ and by the above isomorphism,
an insertion map
$$i_{dS}:\Ber_{M}^{-1}\otimes\Theta_{\Ac}\to \Ac.$$

\begin{definition}
The image of the above insertion map is called the Euler-Lagrange ideal and
denoted $\Ic_{S}$.
The lagrangian variational problem is said to have simplifying histories if its
space of trajectories $T=\{x\in H,\;d_{x}S=0\}$ identifies with
the solution space of the Euler-Lagrange ideal, i.e.,
$$T\cong \Sol_{\Ic_{S}=0}\subset H.$$
\end{definition}

To sum up, a variational problem has simplifying histories if the conditions
imposed on trajectories to define $H$ annihilate the boundary
terms of the integration by part that is used to compute
$d_{x}S$ explicitely (see for example section \ref{simple-example}).

%************************************************************************
\section{Gauge theories and homotopical geometry}
In this last section, we briefly describe a coordinate-free
formulation of gauge theory and of the classical BV formalism using the language
of $\Dc$-schemes of \cite{Beilinson-Drinfeld-Chiral}.
We are inspired here by a huge physical literature,
starting with \cite{Henneaux-Teitelboim} and \cite{Fisch-Henneaux} as general references,
but also \cite{Stasheff2} and \cite{Stasheff1}for some homotopical inspiration, and
\cite{fulp-lada-stasheff},  \cite{Barnich-BV-2010} and \cite{Cattaneo-Felder} for
explicit computations.

For the consistency of this article, we want to insist on the geometrical meaning
of these constructions, continuing to deal with spaces defined by their functor
of points. We will thus use without further comment
\begin{itemize}
\item the language of homotopical geometry, referring to \cite{Toen-HAG-Essen}
for a survey and more references, and
\item the language of pseudo-tensor operations, which we will call here local operations,
referring to \cite{Beilinson-Drinfeld-Chiral},
chapter 1 and 2, for their definition and use.
\end{itemize}
Just recall from Chapter 2 of \cite{Beilinson-Drinfeld-Chiral} the following definition
(we replace everywhere the word pseudo-tensor in this reference by the word local).
\begin{definition}
Let $\Ac$ be a $\Dc$-algebra, $\Mc$ be an $\Ac[\Dc]$-module
and $\Ac^r:=\Ber_{M}\otimes \Ac$ be the corresponding
algebra in the symmetric monoidal category $(\Mod(\Dc^{op}),\otimes^!)$
of right $\Dc$-modules with $\Mc^r:=\Ber_{M}\otimes \Mc$.
A local $2$-ary operation on the $\Ac[\Dc]$-module
$\Mc$ is a morphism
$$\Mc^r\boxtimes \Mc^r\to \Delta_{*}\Mc^r$$
where $\Delta:M\to M\times M$ is the diagonal embedding.
The inner dual of a projective $\Ac[\Dc]$-module $\Mc$ of finite rank is the
$\Ac[\Dc]$-module defined by
$$\Mc^\circ:=\Ber_{M}^{-1}\otimes \Homc_{\Ac[\Dc]}(\Mc,\Ac[\Dc]).$$
\end{definition}

If $\Mc$ is an $\Ac^r[\Dc^{op}]$-module, we denote
$$\Mc^\ell:=\Ber_{M}^{-1}\otimes \Mc$$
the corresponding $\Ac[\Dc]$-module.

\begin{definition}
A variational problem $(\pi:C\to M,H,S\in h(\Ac))$ with simplifying histories
is called a gauge theory.
The kernel of its insertion map
$$i_{dS}:\Theta_{\Ac}^\ell\to \Ac$$
is called the space $\Nc_{S}$ of Noether identities.
Its right version
$$\Nc_{S}^r=\Ber_{M}\otimes \Nc_{S}\subset \Theta_{\Ac}$$
is called the space of Noether gauge symmetries.
\end{definition}

We remark that there is a natural local Lie bracket operation
$$[.,.]:\Theta_{\Ac}\boxtimes \Theta_{\Ac}\to \Delta_{*}\Theta_{\Ac}$$
that plays the role of the Lie bracket between local vector fields.
We refer to Beilinson and Drinfeld's book \cite{Beilinson-Drinfeld-Chiral}
for the following proposition.
\begin{proposition}
The local Lie bracket of vector fields extends naturally to an odd local Poisson bracket on
the dg-$\Ac$-algebra
$$
\Ac_{P}:=
\Sym_{dg}(
[\Theta_{\Ac}^\ell[1]\overset{i_{dS}}{\longrightarrow}\Ac]).
$$
\end{proposition}

One can see the dg-algebra $\Ac_{P}$ as a dg-$\Dc$-space
%(homotopical sheaf on the category of dg-$\Dc$-(super)algebras with values in
%simplicial sets, as in Toen and Vezzosi's work; see \cite{Toen-HAG-Essen})
$$
\begin{array}{ccccc}
P:=\Spec(\Ac_{P}):& dg-\Ac-\Alg & \to 	  &\SSets\\
    & \Rc		&\mapsto & s\Hom_{dg-Alg_{\Dc}}(\Ac_{P},\Rc).
\end{array}
$$
One has then by construction that
$$\pi_{0}(P)\cong \Sol_{\Ic_{S}=0},$$
i.e., the classical (non-homotopical part of) $P$ is exactly the $\Dc$-space of critical points
of the action functional $S$. However, it can have non-trivial higher homotopy,
if the space of Noether identities is non-trivial.

\begin{definition}
The above space $P$ is called the non-proper derived critical space of the given system.
\end{definition}

\begin{corollary}
The natural map
$$\Nc_{S}^r\boxtimes \Nc_{S}^r\to \Delta_{*}\Theta_{\Ac}$$
induced by the bracket on local vector fields always factors through $\Delta_{*}\Nc_{S}^r$
and the natural map
$$\Nc_{S}^r\boxtimes \Ac^r/\Ic_{S}^r\to \Delta_{*}\Ac^r/\Ic_{S}^r$$
is a Lie $\Ac$-algebroid action.
\end{corollary}

Let $\gfk_{S}\to \Nc_{S}$ be a projective $\Ac[\Dc]$-resolution of the Noether identities,
and suppose (to simplify, but this is rarely the case) that the dg-algebra
$$\Bc=\Sym_{dg}([\gfk_{S}[2]\to \Theta_{\Ac}^\ell[1]\to \Ac])$$
is a cofibrant resolution of $\Ac/\Ic_{S}$, whose differential is denoted $d_{KT}$.
From the point of view of derived geometry,
differential forms on this resolution give a definition of the cotangent complex on
the $\Dc$-space $\uSpec_{\Dc}(\Ac/\Ic_{S})$ of critical points of the action functional.
One can think of the derived $\Dc$-stack
$$
\begin{array}{ccccc}
\R T:=\R\uSpec_{\Dc}(\Ac/\Ic_{S}):& dg-\Alg_{\Dc} & \to 	  &\SSets\\
    	& \Rc		&\mapsto & s\Hom_{dg-\Alg_{\Dc}}(\Bc,\Rc)
\end{array}
$$
as a proper solution space for the Euler-Lagrange equation (physicist's language),
or a proper derived critical space. It is a homotopical replacement of the $\Dc$-space
$\uSpec_{\Dc}(\Ac/\Ic_{S})$.

For the following definition, we give a local version of the notion of
$L_{\infty}$-algebroid, whose classical definition can be found in Loday and Vallette's
book \cite{Loday-Vallette-Operads}.
Roughly speeking, a local $L_{\infty}$-algebroid is a representation
of the $L_{\infty}$-operad in the pseudo-tensor category of $\Ac^r[\Dc^{op}]$-modules.
This can be shown to be equivalent to the datum of an inner coderivation on some coalgebra.
Since we are mostly interested in the Chevalley-Eilenberg complex, we will use this
definition.
\begin{definition}
A local $L_{\infty}$-algebroid structure on a graded $\Ac[\Dc]$-module $\Lc$
is given by an $\Ac^r[\Dc^{op}]$-inner coderivation $d$ of degree $1$ of the cocommutative
coassociative inner coalgebra $(C^c(\Lc^r[1]),\Delta)$ over $\Ac^r[\Dc^{op}]$, where $C^c(\Lc^r[1])$ is
the inner exterior algebra of $\Lc^r$ equipped with its natural coalgebra structure
$\Delta:C^c(\Lc^r[1])\to C^c(\Lc^r[1])\wedge C^c(\Lc^r[1])$ defined by
$$\Delta(\gamma_{1})=0$$
and
$$
\Delta(\gamma_{1}\wedge\cdots\wedge \gamma_{n})=
\frac{1}{2}\sum_{k=0}^{n-1}\frac{1}{k!(n-k)!}\sum_{\epsilon\in S_{n}}
\sgn(\epsilon). \gamma_{\epsilon(1)}\wedge\cdots\wedge \gamma_{\epsilon(k)}\bigwedge
\gamma_{\epsilon(k+1)}\wedge\cdots\wedge \gamma_{\epsilon(n)}.
$$
If we suppose that all components of $\Lc^r[1]$ are projective of finite type,
we can dualize $d$ to a derivation $d^\circ$ on $\Sym(\Lc^\circ[1])$
(where $\Lc^\circ$ is the inner dual of $\Lc$).
This gives a definition of the Chevalley-Eilenberg complex of the
$\Ac[\Dc]$-$L_\infty$-algebroid $\Lc$
$$(C(\Lc),d_{CE}):=(\Sym(\Lc^\circ[1]),d^\circ).$$
\end{definition}

\begin{proposition}
Let $\bar{\gfk}_{S}\to \Nc_{S}/\Ic_{S}$ be a projective resolution of
the space $\Nc/\Ic_{S}$ of on-shell Noether identities. There is a natural local
$L_{\infty}$-algebroid structure on $\bar{\gfk}_{S}$. If $\bar{\gfk}_{S}$
is bounded finitely generated, there is a well-defined Chevalley-Eilenberg differential
$$d_{CE}:\Sym(\bar{\gfk}_{S}^\circ[1])\to \Sym(\bar{\gfk}_{S}^\circ[1]).$$
\end{proposition}
\begin{proof}
The result follows from the pseudo-tensor version of theorem
3.5 of \cite{Berger-Moerdijk}, by homotopical transfer of the local Lie bracket on $\Nc_{S}/\Ic_{S}$
to a local $L_{\infty}$-structure on $\gfk_{S}/\Ic_{S}$.
\end{proof}

The aim of the Batalin-Vilkoviski formalism is to define a Poisson differential
graded algebra $(\Ac_{BV},D)$ whose derived $\Dc$-stack
$$
\begin{array}{ccccc}
\R\uSpec_{\Dc}(\Ac_{BV},D):& dg-\Alg_{\Dc} & \to 	  &\SSets\\
    	& \Bc		&\mapsto & s\Hom_{dg-\Alg_{\Dc}}((\Ac_{BV},D),\Bc)
\end{array}
$$
is a kind of homotopical space of leaves
$$\R\Spec(\Ac/\Ic_{S})\underset{\Lb}{/}(\Nc_{S}^r/\Ic_{S}^r)$$
of the derived critical space $\R\Spec(\Ac/\Ic_{S})$ by the ``foliation by gauge orbits''
defined by the Lie algebroid
$\Nc_{S}^r/\Ic_{S}^r$ of on-shell gauge symmetries.
The differential $D$ is essentially obtained, under additional hypothesis, by combining in a neat way
the above Chevalley-Eilenberg differential $d_{CE}$ for the $\gfk_{S}/\Ic_{S}$-module $\Ac/\Ic_{S}$
with the Koszul-Tate differential $d_{KT}$ on the cofibrant resolution $\Bc$ of $\Ac/\Ic_{S}$.
This neat combination could be done,
for example if $\gfk_{S}\to \Nc_{S}$ was a projective resolution of the
Noether identities, by extending the local action map
$$\Nc_{S}^r\boxtimes \Ac^r/\Ic_{S}^r\to \Delta_{*}\Ac^r/\Ic_{S}^r$$
to an $\infty$-action
$$\gfk_{S}^r\boxtimes \Bc^r\to \Delta_{*}\Bc^r$$
of the local $L_{\infty}$-$\Ac$-algebroid $\gfk_{S}$ on the resolution $\Bc$ of $\Ac/\Ic_{S}$,
and taking the total complex of the associated Chevalley-Eilenberg complex
(see \cite{Beilinson-Drinfeld-Chiral}, section 1.4.5)
$$(C(\gfk_{S},\Bc),d_{CE}).$$
Remark however that this object is only an $\R_{M}$-algebra and
one would like to have an $\Ac$-algebra here, by replacing the Chevalley-Eilenberg
complex by an \emph{inner} version of it.
The existence of the inner Chevalley-Eilenberg complex is only given under very strong
finite-dimension conditions, that are not fulfilled in the above construction.
The essentially role of the Batalin-Vilkovisky construction is to give a systematic
way to fill the above conceptual gap, by using smaller generating spaces of
Noether symmetries.

Indeed, remark that the left hand side of the natural map
$$\wedge^2\Theta_{\Ac}^\ell\to \Nc.$$
is not a finitely generated $\Ac[\Dc]$-module (contrary
to what would occur in a finite dimensional geometric situation) for the
same reason that $\Dc\otimes_{\Oc}\Dc$ is not $\Dc$-coherent.
This shows that it is hard to find an $\Ac[\Dc]$-finitely generated off-shell projective
resolution $\gfk_{S}$ of the space $\Nc_{S}$ of Noether identities. Another problem
is that such a projective resolution does not give, in general, a cofibrant resolution of
the $\Ac$-algebra $\Ac/\Ic_S$ because the differential graded symmetric algebra
functor $\Sym_{dg-\Ac}$ is not always exact.
This motivates the following construction, that is also useful for computational purposes.

Generating spaces of Noether gauge symmetries can be defined by adapting Tate's construction
\cite{Tate4} to the local context. We are inspired here by Stasheff's paper \cite{Stasheff2}.
\begin{definition}
A generating space of Noether identities is a tuple $(\gfk_{S},\Ac_{n},i_{n})$
composed of
\begin{enumerate}
\item a negatively graded projective $\Ac[\Dc]$-module $\gfk_{S}$,
\item a negatively indexed family $\Ac_{n}$ of dg-$\Ac$-algebras with $\Ac_{0}=\Ac$, and
\item for each $n\leq -1$, an $\Ac[\Dc]$-linear morphism $i_{n}:\gfk_{S}^{n+1}\to Z^n\Ac_{n}$ to
the $n$-cycles of $\Ac_{n}$,
\end{enumerate}
such that if one extends $\gfk_{S}$ by setting $\gfk_{S}^1=\Theta_{\Ac}^\ell$
and if one sets
$$i_{0}=i_{dS}:\Theta_{\Ac}^\ell\to \Ac,$$
\begin{enumerate}
\item one has for all $n\leq 0$ an equality
$$
\Ac_{n-1}=\Sym_{\Ac_{n}}([\gfk_{S}^{n+1}[-n+1]\otimes_{\Ac}\Ac_{n}\overset{i_{n}}{\to} \underset{0}{\Ac_{n}}]),
$$
\item the natural projection map
$$\Ac_{KT}:=\limind \Ac_{n}\to \Ac/\Ic_{S}$$
is a cofibrant resolution, called the Koszul-Tate algebra,
whose differential is denoted $d_{KT}$.
\end{enumerate}
\end{definition}

\begin{lemma}
The complex
$$\Pc_{KT}:=[\Ac_{KT}/(\Theta_{\Ac}^\ell)]^{\leq -2}$$
of components of degree smaller than $-2$ in the quotient algebra of
the Koszul-Tate algebra by the ideal of local vector fields maps to the space
$\Nc_{S}$ of Noether identities, and surjects onto $\Nc_{S}/\Ic_{S}$.
Its underlying graded module is
$$\Pc_{KT}=\Sym_{g}(\gfk_{S}[2]).$$
The inclusion $\gfk_{S}[2]\subset \Pc_{KT}$ induces a degree $1$ map
$$\widetilde{d}:\gfk_{S}[1]\to \gfk_{S}.$$
\end{lemma}
\begin{proof}
The first statements follow from the definition of the Koszul-Tate algebra.
The inclusion and projection of homogeneous components induce
natural maps
$$\gfk_{S}\to \Pc_{KT}\textrm{ and }\Pc_{KT}\to \gfk_{S},$$
that can be composed with the differential on $\Pc_{KT}$ to define $\widetilde{d}$.
\end{proof}

\begin{definition}
One says that a generating space of Noether identities is
\begin{enumerate}
\item strongly regular if the graded space $\gfk_{S}$ is bounded with finitely generated projective
components,
\item on-shell complete if the natural map $\gfk_{S}\to \Nc_{S}$ induces a projective resolution
$$\gfk_{S}/\Ic_{S}\to \Nc_{S}/\Ic_{S}$$
of $\Nc_{S}/\Ic_{S}$ as an $\Ac/\Ic_{S}$-module, with differential induced by $\widetilde{d}$.
\item on-shell algebraically complete if the natural map $\Pc_{KT}\to \Nc_{S}$ induces
a projective resolution
$$\Pc_{KT}/\Ic_{S}\to \Nc_{S}/\Ic_{S}.$$
\end{enumerate}
\end{definition}

Since $\Dc\otimes_{\Oc}\Dc$ is usually not $\Dc$-coherent,
the higher degree homogeneous components of the Koszul-Tate algebra (e.g,
the component $\wedge^2\Theta_{\Ac}^\ell$, that corresponds
to trivial Noether identities) are usually not of finite type over $\Ac[\Dc]$,
even under the strongly regular hypothesis.
This problem is specific to local field theory and does not occur in finite dimensional geometry:
the graded algebras in play are of finite type, but their higher homogeneous components are not
finitely generated modules.

\begin{proposition}
Let $\gfk_{S}$ be an on-shell complete (respectively, on-shell algebraically complete) generating
space of Noether identities.
There is a natural local $L_{\infty}$-algebroid structure on $\bar{\gfk}_{S}:=\gfk_{S}/\Ic_{S}$
(resp. $\Pc_{KT}/\Ic_{S}$).
If $\gfk_{S}$ is on-shell complete and strongly regular,
there is a well-defined degree $1$ map
$$d_{CE}:\Sym(\gfk_{S}^\circ[1])\to \Sym(\gfk_{S}^\circ[1])$$
induced by the $L_{\infty}$-algebroid structure on $\bar{\gfk}_{S}$.
\end{proposition}
\begin{proof}
The condition of on-shell completeness means that the natural map
$$\bar{\gfk}_{S}\to \Nc_{S}/\Ic_{S}$$
is a projective resolution. The result follows from the pseudo-tensor version of theorem
3.5 of \cite{Berger-Moerdijk}, by homotopical transfer of the local Lie bracket on $\Nc_{S}/\Ic_{S}$
to a local $L_{\infty}$-structure on $\bar{\gfk}_{S}$.
The extension of $d_{CE}$ to $\Sym(\gfk_{S}^\circ[1])$ is given by the fact
that $\gfk_{S}^\circ\to\gfk_{S}^\circ/\Ic_{S}$ is componentwise surjective and $\gfk_{S}^\circ$
has projective components.
\end{proof}

\begin{definition}
One says that the gauge symmetries close off-shell if there exists a
generating space of Noether gauge symmetries $\gfk_{S}$ that is finitely
$\Ac[\Dc]$-generated, and whose
image $\bar{\gfk}_{S}$ in $\Theta_{\Ac}^\ell$ admits a bracket
$$\bar{\gfk}_{S}^r\boxtimes \bar{\gfk}_{S}^r\to \Delta_{*}\bar{\gfk}^r_{S}$$
induced by the local bracket on vector fields. One says that the theory is
$N$-reducible is there exists a minimal strongly regular generating space of length $N$.
In particular, $0$-reducible gauge theories are called irreducible gauge
theories.
\end{definition}

\begin{definition}
Let $\gfk_{S}$ be a strongly regular generating space of the Noether gauge symmetries.
Such a generating space is called a space of antighosts of the gauge theory.
The inner dual space (that is well-defined because of the strong regularity
hypothesis)
$$\gfk_{S}^\circ:=\Ber_{M}^{-1}\otimes\Homc_{\Ac[\Dc]}(\gfk_{S},\Ac[\Dc])$$
is called the space of ghosts.
\end{definition}

\begin{theorem}
Let $\gfk_{S}$ be a strongly regular generating space of Noether identities and $\gfk_{S}^\circ$
its inner dual. The bigraded algebra
$$
\Ac_{BV,bigrad}:=
\Sym_{bigrad}\left(\left[
\begin{array}{ccccc}
\gfk_{S}[2] &\oplus & \Ber_{M}^{-1}\otimes\Theta_{\Ac}[1] &\oplus & \Ac\\
&&&&\oplus\\
&&&& ^t\gfk_{S}^{\circ}[-1]
\end{array}\right]\right),
$$
where $^t\gfk_{S}^{\circ}$ is the vertical chain graded space associated to $\gfk_{S}^\circ$,
is equiped with a natural local bracket
$$\{.,.\}:\Ac_{BV,bigrad}^r\boxtimes \Ac_{BV,bigrad}^r\to \Delta_{*}\Ac_{BV,bigrad}^r$$
called the antibracket.
%If moreover the gauge symmetries close off-shell, one can identify $\Ac_{BV,g}$
%with the underlying bigraded algebra of a bi-$dg$-algebra
%$$
%\Ac_{BV,bi-dg}:=
%\Sym_{dg}([\gfk_{S}[2]\to \Ber_{M}^{-1}\otimes\Theta_{\Ac}[1]\overset{i_{dS}}{\to} \Ac])
%\otimes_{\Ac} C(\gfk_{S},d),
%$$
%whose differentials are the Koszul-Tate differential $\delta$
%and the Chevalley-Eilenberg differential $d$ for the $\infty$-algebroid $\gfk_{S}$.
\end{theorem}
\begin{proof}
There is a natural duality pairing
$$\langle.,.\rangle:\gfk_{S}^r\boxtimes (\gfk_{S}^\circ)^r\to \Delta_{*}\Ac^r$$
between antighosts and ghosts.
Similarly to the finite-dimensional case treated in
\cite{Kostant-Sternberg}, this duality and the isomorphism
$$
\gr^\bullet \Cliff_{inner}(\gfk_{S}^r[2]\oplus (\gfk_{S}^\circ)^r[-1],\langle.,.\rangle)
\cong \Sym(\gfk_{S}^r[2]\oplus (\gfk_{S}^\circ)^r[-1])
$$
induce a local Poisson bracket on
$$
\Sym(\gfk_{S}[2]\oplus \gfk_{S}^\circ[-1])\cong
\Sym(\gfk_{S}[2])\otimes \Sym(\gfk_{S}^\circ[-1]).
$$
Combining this with the local Schouten-Nijenhuis bracket
$$
\{.,.\}:(\wedge^*\Theta_{\Ac})\boxtimes (\wedge^*\Theta_{\Ac})\to \Delta_{*}(\wedge^*\Theta_{\Ac}),
$$
one gets a local Poisson bracket on the bigraded algebra
$$
\Ac_{BV,bigrad}:=
\Sym_{bigrad}\left(\left[
\begin{array}{ccccc}
\gfk_{S}[2] &\oplus & \Ber_{M}^{-1}\otimes\Theta_{\Ac}[1] &\oplus & \Ac\\
&&&&\oplus\\
&&&& ^t\gfk_{S}^{\circ}[-1]
\end{array}\right]\right).
$$
%If the gauge symmetries close off-shell, the part $\Sym(\gfk_{S}^\circ)$ of this algebra identifies
%with the inner Chevalley-Eilenberg complex
%$$(C(\gfk_{S}),d)$$
%of the $L_{\infty}$-algebra $\gfk_{S}$. One can thus consider the underlying graded algebra
%$\Ac_{BV,g}$ of $\Ac_{BV,bigrad}$ as a differential
%graded algebra given by
%$$
%\Ac_{BV,dg}:=
%\Sym_{dg}([\gfk_{S}[2]\to \Ber_{M}^{-1}\otimes\Theta_{\Ac}[1]\overset{i_{dS}}{\to} \Ac])
%\otimes_{\Ac} C(\gfk_{S},d).
%$$
%The two differentials that appear here are the Koszul-Tate differential for
%the ideal $\Ic_{S}\subset \Ac$ and the Chevalley-Eilenberg differential for the
%$\infty$-algebroid $\gfk_{S}$, respectively.
\end{proof}

\begin{definition}
Let $\gfk_{S}$ be a strongly regular generating space of Noether identities.
The corresponding BV algebra is the local Poisson algebra
$\Ac_{BV,bigrad}$. A solution to the classical master equation is an
$S_{cm}\in h(\Ac_{BV,bigrad})$ such that
\begin{enumerate}
\item the degree $(0,0)$ component of $S_{cm}$ is $S$,
\item a component of $S_{cm}$, denoted $S_{KT}$, induces
the Koszul-Tate differential $d_{KT}=\{S_{KT},.\}$
on antifields of degrees $(k,0)$, and
\item the master equation
$$\{S_{cm},S_{cm}\}=0$$
(meaning $D^2=0$ for $D=\{S_{cm},.\}$) is fulfilled in $h(\Ac_{BV,bigrad})$.
%\newcounter{enumi_saved}
%\setcounter{enumi_saved}{\value{enumi}}
%\end{enumerate}
%Suppose that the gauge symmetries close off-shell. A non-degenerate solution to the
%classical master equation is a solution $S_{cm}\in h(\Ac_{BV,dg})$ such that, moreover,
%\begin{enumerate}
%\setcounter{enumi}{\value{enumi_saved}}
%\item a component of $S_{cm}$, denoted $S_{CE}$, induces the Cartan-Eilenberg
%differential on ghosts of degree $(0,k)$.
\end{enumerate}
\end{definition}

One can also add some conditions on $S_{cm}$ related to the on-shell Chevalley-Eilenberg
differential $d_{CE}$.

The main theorem of homological perturbation theory, given in a physical language
in \cite{Henneaux-Teitelboim}, chapter 17  (DeWitt indices), can be formulated
in our case by the following.

\begin{theorem}
Let $\gfk_{S}$ be a strongly regular and on-shell complete generating space
of Noether symmetries. 
There exists a solution to the corresponding classical master equation.
\end{theorem}

%We remark that the main interest of the BV formalism is that it also
%applies to theories that do not close off-shell.

As explained above, the space $\R\uSpec_{\Dc}(\Ac_{BV},D)$ can be thought as a kind of
homotopical space of leaves
$$\R\Spec(\Ac/\Ic_{S})\underset{\Lb}{/}\Nc_{S}^r$$
of the foliation induced by the Noether gauge symmetries $\Nc_{S}^r$
on the derived critical space $\R\uSpec_{\Dc}(\Ac/\Ic_{S})$.
It is naturally equipped with a homotopical Poisson structure, which gives a nice starting point for quantization. 
%The condition imposed on $S_{cm}$ related to the Cartan-Eilenberg
%differential must then be modified (because the local Lie bracket on Noether
%gauge symmetries, and consequently their Cartan-Eilenberg complex, is only
%well-defined on-shell) but the main theorem continues to
%hold.

%************************************************************************
\section*{Acknowledgements}

The author thanks Michael Baechtold, Glenn Barnich, Henrique Bursztyn, Damien Calaque,
Alberto Cattaneo, Denis-Charles Cisinski, Thiaggo Drummond, Giovanni Felder, Gregory Ginot,
Patrick Iglesias-Zemmour, Joseph Krasilsh'chik, Joan Mill\`es,
Giovanni Morando, Jim Stasheff, Pierre Schapira, Mathieu Stienon, Bertrand Toen, Bruno Vallette,
Alexander Verbovetsky, Luca Vitagliano, and Alan Weinstein for useful discussions.
He thanks the referees for their careful reading of a first version of the paper.
The author also thanks the organizers of the summer 2008 diffiety school,
Jussieu's Mathematical Institute, the University of Paris 6, and the IMPA
in Rio, for providing him with excellent working conditions for the preparation of
this article.

% *************************
% Appel de la bibliographie
% *************************
\bibliographystyle{alpha}
%\bibliography{/share/nfs/users/imj-ao/fpaugam/travail/fred}
%\bibliography{/home/visitor/fpaugam/travail/fred}
\bibliography{$HOME/Documents/travail/fred}

\newcommand{\etalchar}[1]{$^{#1}$}
\def\cprime{$'$} \def\cprime{$'$} \def\cprime{$'$} \def\cprime{$'$}
  \def\cprime{$'$} \def\cprime{$'$}
\begin{thebibliography}{BCD{\etalchar{+}}99}

\bibitem[AGV73]{SGA4}
M.~Artin, A.~Grothendieck, and J.-L. Verdier.
\newblock {\em Th\'eorie des topos et cohomologie \'etale des sch\'emas. {T}ome
  3}.
\newblock Springer-Verlag, Berlin, 1973.
\newblock S\'eminaire de G\'eom\'etrie Alg\'ebrique du Bois-Marie 1963--1964
  (SGA 4), Dirig\'e par M. Artin, A. Grothendieck et J. L. Verdier. Avec la
  collaboration de P. Deligne et B. Saint-Donat, Lecture Notes in Mathematics,
  Vol. 305.

\bibitem[Bar10]{Barnich-BV-2010}
G.~Barnich.
\newblock A note on gauge systems from the point of view of {L}ie algebroids.
\newblock {\em ArXiv e-prints}, oct 2010.

\bibitem[BCD{\etalchar{+}}99]{Vinogradov-Krasilshchik}
A.~V. Bocharov, V.~N. Chetverikov, S.~V. Duzhin, N.~G. Khor{\cprime}kova, I.~S.
  Krasil{\cprime}shchik, A.~V. Samokhin, Yu.~N. Torkhov, A.~M. Verbovetsky, and
  A.~M. Vinogradov.
\newblock {\em Symmetries and conservation laws for differential equations of
  mathematical physics}, volume 182 of {\em Translations of Mathematical
  Monographs}.
\newblock American Mathematical Society, Providence, RI, 1999.
\newblock Edited and with a preface by Krasil{\cprime}shchik and Vinogradov,
  Translated from the 1997 Russian original by Verbovetsky [A. M.
  Verbovetski{\u\i}] and Krasil{\cprime}shchik.

\bibitem[BD04]{Beilinson-Drinfeld-Chiral}
Alexander Beilinson and Vladimir Drinfeld.
\newblock {\em Chiral algebras}, volume~51 of {\em American Mathematical
  Society Colloquium Publications}.
\newblock American Mathematical Society, Providence, RI, 2004.

\bibitem[BH08]{Baez-diffeo}
John~C. Baez and Alexander~E. Hoffnung.
\newblock Convenient categories of smooth spaces, 2008.

\bibitem[BL77]{Bernstein-Leites}
I.~N. Bern{\v{s}}te{\u\i}n and D.~A. Le{\u\i}tes.
\newblock Integral forms and the {S}tokes formula on supermanifolds.
\newblock {\em Funkcional. Anal. i Prilo\v zen.}, 11(1):55--56, 1977.

\bibitem[BM03]{Berger-Moerdijk}
Clemens Berger and Ieke Moerdijk.
\newblock Axiomatic homotopy theory for operads.
\newblock {\em Comment. Math. Helv.}, 78(4):805--831, 2003.

\bibitem[CF01]{Cattaneo-Felder}
Alberto~S. Cattaneo and Giovanni Felder.
\newblock Poisson sigma models and deformation quantization.
\newblock {\em Modern Phys. Lett. A}, 16(4-6):179--189, 2001.
\newblock Euroconference on Brane New World and Noncommutative Geometry
  (Torino, 2000).

\bibitem[Che77]{Chen-1977}
Kuo~Tsai Chen.
\newblock Iterated path integrals.
\newblock {\em Bull. Amer. Math. Soc.}, 83(5):831--879, 1977.

\bibitem[Con94]{Connes2}
Alain Connes.
\newblock {\em Noncommutative geometry}.
\newblock Academic Press Inc., San Diego, CA, 1994.

\bibitem[Cos10]{Costello}
Kevin Costello.
\newblock Renormalization and effective field theory, 2010.

\bibitem[DeW03]{Dewitt}
Bryce DeWitt.
\newblock {\em The global approach to quantum field theory. {V}ol. 1, 2},
  volume 114 of {\em International Series of Monographs on Physics}.
\newblock The Clarendon Press Oxford University Press, New York, 2003.

\bibitem[DF99]{classical-fields-deligne-freed}
Pierre Deligne and Daniel~S. Freed.
\newblock Classical field theory.
\newblock In {\em Quantum fields and strings: a course for mathematicians,
  {V}ol. 1, 2 ({P}rinceton, {NJ}, 1996/1997)}, pages 137--225. Amer. Math.
  Soc., Providence, RI, 1999.

\bibitem[DM99]{notes-on-supersymmetry}
Pierre Deligne and John~W. Morgan.
\newblock Notes on supersymmetry (following {J}oseph {B}ernstein).
\newblock In {\em Quantum fields and strings: a course for mathematicians,
  {V}ol. 1, 2 ({P}rinceton, {NJ}, 1996/1997)}, pages 41--97. Amer. Math. Soc.,
  Providence, RI, 1999.

\bibitem[FH90]{Fisch-Henneaux}
Jean M.~L. Fisch and Marc Henneaux.
\newblock Homological perturbation theory and the algebraic structure of the
  antifield-antibracket formalism for gauge theories.
\newblock {\em Comm. Math. Phys.}, 128(3):627--640, 1990.

\bibitem[FLS02]{fulp-lada-stasheff}
Ron Fulp, Tom Lada, and Jim Stasheff.
\newblock Noether's variational theorem {II} and the {BV} formalism.
\newblock {\em arXiv}, 2002.

\bibitem[Gro60a]{Gro1}
A.~Grothendieck.
\newblock \'{E}l\'ements de g\'eom\'etrie alg\'ebrique. {I}. {L}e langage des
  sch\'emas.
\newblock {\em Inst. Hautes \'Etudes Sci. Publ. Math.}, 4:228, 1960.

\bibitem[Gro60b]{EGAI}
A.~Grothendieck.
\newblock \'{E}l\'ements de g\'eom\'etrie alg\'ebrique. {I}. {L}e langage des
  sch\'emas.
\newblock {\em Inst. Hautes \'Etudes Sci. Publ. Math.}, 1-4(4):228, 1960.

\bibitem[Gro86]{Gromov}
Mikhael Gromov.
\newblock {\em Partial differential relations}, volume~9 of {\em Ergebnisse der
  Mathematik und ihrer Grenzgebiete (3) [Results in Mathematics and Related
  Areas (3)]}.
\newblock Springer-Verlag, Berlin, 1986.

\bibitem[HT92]{Henneaux-Teitelboim}
Marc Henneaux and Claudio Teitelboim.
\newblock {\em Quantization of gauge systems}.
\newblock Princeton University Press, Princeton, NJ, 1992.

\bibitem[IZ99]{piz}
Patrick Iglesias~Zeimour.
\newblock {\em Diffeology}, volume~1 of {\em web book}.
\newblock http://math.huji.ac.il/~piz/Site/The Book/The Book.html, 1999.

\bibitem[Kas03]{Kashiwara-D-modules}
Masaki Kashiwara.
\newblock {\em {$D$}-modules and microlocal calculus}, volume 217 of {\em
  Translations of Mathematical Monographs}.
\newblock American Mathematical Society, Providence, RI, 2003.
\newblock Translated from the 2000 Japanese original by Mutsumi Saito, Iwanami
  Series in Modern Mathematics.

\bibitem[KS87]{Kostant-Sternberg}
Bertram Kostant and Shlomo Sternberg.
\newblock Symplectic reduction, {BRS} cohomology, and infinite-dimensional
  {C}lifford algebras.
\newblock {\em Ann. Physics}, 176(1):49--113, 1987.

\bibitem[KS90]{Kashiwara-Schapira-sheaves-on-manifolds}
Masaki Kashiwara and Pierre Schapira.
\newblock {\em Sheaves on manifolds}, volume 292 of {\em Grundlehren der
  Mathematischen Wissenschaften [Fundamental Principles of Mathematical
  Sciences]}.
\newblock Springer-Verlag, Berlin, 1990.
\newblock With a chapter in French by Christian Houzel.

\bibitem[KS06a]{Kashiwara-Schapira-categories-and-sheaves}
Masaki Kashiwara and Pierre Schapira.
\newblock {\em Categories and sheaves}, volume 332 of {\em Grundlehren der
  Mathematischen Wissenschaften [Fundamental Principles of Mathematical
  Sciences]}.
\newblock Springer-Verlag, Berlin, 2006.

\bibitem[KS06b]{Kosmann-Schwarzbach}
Yvette Kosmann-Schwarzbach.
\newblock {\em Les th\'eor\`emes de {N}oether}.
\newblock Histoire des Math\'ematiques (Palaiseau). [History of Mathematics
  (Palaiseau)]. \'Editions de l'\'Ecole Polytechnique, Palaiseau, second
  edition, 2006.
\newblock Invariance et lois de conservation au XX${\sp{{}}{\rm{e}}}$
  si{\`e}cle. [Invariance and conservation laws in the twentieth century], With
  a translation of the original article ``Invariante Variationsprobleme'', With
  the collaboration of Laurent Meersseman.

\bibitem[KV98]{Krasilshchik-Verbovetsky-1998}
Joseph Krasil'shchik and Alexander Verbovetsky.
\newblock Homological methods in equations of mathematical physics.
\newblock {\em arXiv}, 1998.

\bibitem[Lot90]{Lott-torsion-supergeometry}
John Lott.
\newblock Torsion constraints in supergeometry.
\newblock {\em Comm. Math. Phys.}, 133(3):563--615, 1990.

\bibitem[LV10]{Loday-Vallette-Operads}
Jean-Louis Loday and Bruno Vallette.
\newblock Algebraic operads.
\newblock {\em (book in preparation)}, 2010.

\bibitem[Lyc93]{Lychagin93}
V.~Lychagin.
\newblock Quantizations of braided differential operators, 1993.

\bibitem[Nes03]{Nestruev}
Jet Nestruev.
\newblock {\em Smooth manifolds and observables}, volume 220 of {\em Graduate
  Texts in Mathematics}.
\newblock Springer-Verlag, New York, 2003.
\newblock Joint work of A. M. Astashov, A. B. Bocharov, S. V. Duzhin, A. B.
  Sossinsky, A. M. Vinogradov and M. M. Vinogradov, Translated from the 2000
  Russian edition by Sossinsky, I. S. Krasil{\cprime}schik and Duzhin.

\bibitem[Pen83]{Penkov}
I.~B. Penkov.
\newblock {${\mathcal D}$}-modules on supermanifolds.
\newblock {\em Invent. Math.}, 71(3):501--512, 1983.

\bibitem[Rit66]{Ritt}
Joseph~Fels Ritt.
\newblock {\em Differential algebra}.
\newblock Dover Publications Inc., New York, 1966.

\bibitem[Sch94]{Schneiders}
Jean-Pierre Schneiders.
\newblock An introduction to {$\Dc$}-modules.
\newblock {\em Bull. Soc. Roy. Sci. Li\`ege}, 63(3-4):223--295, 1994.
\newblock Algebraic Analysis Meeting (Li{\`e}ge, 1993).

\bibitem[Smi66]{Smith-1966}
J.~Wolfgang Smith.
\newblock The de {R}ham theorem for general spaces.
\newblock {\em T\^ohoku Math. J. (2)}, 18:115--137, 1966.

\bibitem[Sou97]{Souriau}
J.-M. Souriau.
\newblock {\em Structure of dynamical systems}, volume 149 of {\em Progress in
  Mathematics}.
\newblock Birkh\"auser Boston Inc., Boston, MA, 1997.
\newblock A symplectic view of physics, Translated from the French by C. H.
  Cushman-de Vries, Translation edited and with a preface by R. H. Cushman and
  G. M. Tuynman.

\bibitem[Sta97]{Stasheff2}
Jim Stasheff.
\newblock Deformation theory and the {B}atalin-{V}ilkovisky master equation.
\newblock In {\em Deformation theory and symplectic geometry ({A}scona, 1996)},
  volume~20 of {\em Math. Phys. Stud.}, pages 271--284. Kluwer Acad. Publ.,
  Dordrecht, 1997.

\bibitem[Sta98]{Stasheff1}
Jim Stasheff.
\newblock The (secret?) homological algebra of the {B}atalin-{V}ilkovisky
  approach.
\newblock In {\em Secondary calculus and cohomological physics ({M}oscow,
  1997)}, volume 219 of {\em Contemp. Math.}, pages 195--210. Amer. Math. Soc.,
  Providence, RI, 1998.

\bibitem[Tat57]{Tate4}
John Tate.
\newblock Homology of {N}oetherian rings and local rings.
\newblock {\em Illinois J. Math.}, 1:14--27, 1957.

\bibitem[Toe]{Toen-HAG-Essen}
Bertrand Toen.
\newblock {From homotopical algebra to homotopical algebraic geometry: lectures
  in Essen}.

\bibitem[TV08a]{Toen-Vaquie-CACA}
Bertrand To{\"e}n and Michel Vaqui{\'e}.
\newblock Alg\'ebrisation des vari\'et\'es analytiques complexes et
  cat\'egories d\'eriv\'ees.
\newblock {\em Math. Ann.}, 342(4):789--831, 2008.

\bibitem[TV08b]{Toen-Vezzosi-HAG-II}
Bertrand To{\"e}n and Gabriele Vezzosi.
\newblock Homotopical algebraic geometry. {II}. {G}eometric stacks and
  applications.
\newblock {\em Mem. Amer. Math. Soc.}, 193(902):x+224, 2008.

\bibitem[TV09]{Toen-Vaquie}
Bertrand To{\"e}n and Michel Vaqui{\'e}.
\newblock Au-dessous de {${\rm Spec}\,\Bbb Z$}.
\newblock {\em J. K-Theory}, 3(3):437--500, 2009.

\bibitem[Vin01]{Vinogradov}
A.~M. Vinogradov.
\newblock {\em Cohomological analysis of partial differential equations and
  secondary calculus}, volume 204 of {\em Translations of Mathematical
  Monographs}.
\newblock American Mathematical Society, Providence, RI, 2001.
\newblock Translated from the Russian manuscript by Joseph
  Krasil{\cprime}shchik.

\end{thebibliography}

\end{document}